%% file: main.tex
\documentclass[11pt]{llncs}
\usepackage{thm-restate}
\usepackage{mathtools}
\usepackage{amsmath}
\usepackage{amssymb}
\usepackage[table, dvipsnames]{xcolor}
\usepackage{adjustbox}
\usepackage{array, makecell}
\usepackage{framed}
\usepackage{setspace}
\usepackage{graphicx}
\usepackage{bbm}
\usepackage{soul}
\usepackage[normalem]{ulem}

\newcommand{\opt}{{\mathsf{OPT}}}
\usepackage{amsfonts}

\usepackage[breaklinks=true]{hyperref}
\hypersetup{colorlinks=true,%
            citebordercolor={.6 .6 .6},linkbordercolor={.6 .6 .6},%
citecolor=blue,urlcolor=black,linkcolor=blue,pagecolor=black}
\usepackage[nameinlink]{cleveref}
\Crefname{algocf}{Algorithm}{Algorithms}
\crefname{algocfline}{line}{lines}
\Crefname{invariant}{Invariant}{Invariants}
\Crefname{claim}{Claim}{Claims}
\Crefname{clm}{Claim}{Claims}
\Crefname{figure}{Figure}{Figures}
\Crefname{subclaim}{Subclaim}{Subclaims}

\usepackage{epsfig}

\usepackage{mathrsfs}
\usepackage{xspace}
\usepackage{soul}
\usepackage{latexsym}
\usepackage{bbm}

\usepackage{framed}

\usepackage[dvipsnames]{xcolor}
\definecolor{DarkGray}{rgb}{0.66, 0.66, 0.66}
\definecolor{DarkPowderBlue}{rgb}{0.0, 0.2, 0.6}
\definecolor{fluorescentyellow}{rgb}{0.8, 1.0, 0.0}

\usepackage[ruled,vlined,linesnumbered,algonl]{algorithm2e}
\SetEndCharOfAlgoLine{}

\usepackage{breakcites}
\usepackage{scalerel}[2016/12/29] 
\usepackage{multirow}
\usepackage{tikz}

\SetCommentSty{mycommfont}
\DontPrintSemicolon

\usepackage{newfloat}
\DeclareFloatingEnvironment[fileext=lop]{Algorithm}

\newcommand{\X}{\mathcal{X}}
\newcommand{\veps}{\varepsilon}

\newcommand{\rad}{{\mathsf {radius}}}

\newcommand{\I}{\mathcal{I}}
\newcommand{\cO}{\mathcal{O}}

\newcommand{\B}{\mathcal{B}}

\newcommand{\kcen}{\textsc{$k$-center}\xspace}
\newcommand{\capcen}{\textsc{Capacitated $k$-center}\xspace}
\newcommand{\kmed}{\textsc{$k$-median}\xspace}
\newcommand{\capmed}{\textsc{Capacitated $k$-median}\xspace}

\newcommand{\kmean}{\textsc{$k$-means}\xspace}
\newcommand{\ksum}{\textsc{$k$-sumRadii}\xspace}
\newcommand{\polylog}{\mathsf{polylog}}
\newcommand{\poly}{{\mathsf{poly}}}
\newcommand{\capsum}{\textsc{Capacitated $k$-sumRadii}\xspace}
\newcommand{\ETH}{\textsc{ETH}\xspace}
\newcommand{\vc}{\textsc{Vertex Cover}\xspace}
\renewcommand{\eth}{\textsc{Exponential Time Hypothesis}\xspace}

\definecolor{mycolor}{rgb}{1, 0.0, 0.0}
\newcommand{\dishant}[1]{\textcolor{mycolor}{}}

\newcommand{\inv}[1]{{\bf Invariant #1}}

\allowdisplaybreaks
\newcommand{\cE}{{\cal E}}
\newcommand{\cI}{{\cal I}}

\begin{document}

\title{FPT Approximation for Capacitated Sum of Radii}
%
%
\author{Ragesh Jaiswal, Amit Kumar and Jatin Yadav}
%
%
%
\institute{Department of Computer Science and Engineering, IIT Delhi} 

{\def\addcontentsline#1#2#3{}\maketitle}

\begin{abstract}
    We consider the capacitated clustering problem in general metric spaces where the goal is to identify $k$ clusters and minimize the sum of the radii of the clusters (we call this the  \capsum problem). We are interested in fixed-parameter tractable (FPT) approximation algorithms where the running time is of the form $f(k) \cdot \poly(n)$, where $f(k)$ can be an exponential function of $k$ and $n$ is the number of points in the input.  In the uniform capacity case, Bandyapadhyay et al. recently gave a $4$-approximation algorithm  for this problem. Our first result improves this to an FPT $3$-approximation and extends to a constant factor approximation for any $L_p$ norm of the cluster radii. In the general capacities version, Bandyapadhyay et al. gave an FPT $15$-approximation algorithm. We extend their framework to give an FPT $(4 + \sqrt{13})$-approximation algorithm for this problem. Our framework relies on a novel idea of identifying approximations to optimal clusters by carefully pruning points from an initial candidate set of points. This is in contrast to prior results that rely on guessing suitable points and building balls of appropriate radii around them.      
    
    On the hardness front, 
    we show  that assuming the Exponential Time Hypothesis, there is a constant $c > 1$ such that any $c$-approximation algorithm for the non-uniform capacity version of this problem requires running time $2^{\Omega \left(\frac{k}{polylog(k)} \right)}$. 
    \end{abstract}
\input{1-Introduction}

\input{2-Uniform-capacity}

\input{3-Non-uniform-capacity}

\input{4-Hardness}
\bibliographystyle{plainurl}
\bibliography{refs}

\end{document}

%% file: 1-Introduction.tex
\section{Introduction}
Center-based clustering problems are important data processing tasks. Given a metric $D$ on a set of $n$ points $\X$ and a parameter $k$, the goal here is to partition the set of points into $k$ {\it clusters}, say $C_1, \ldots, C_k$,  and assign the points in each cluster to a corresponding {\it cluster center}, say $c_1, \ldots, c_k$, respectively. Some of the most widely studied objective functions are given by the  \kcen, the \kmed and the \kmean problems. The \kcen problem seeks to find a clustering such that the maximum radius of a cluster is minimized. Here, radius of a cluster is defined as the farthest distance between the center of the cluster and a point in the same cluster. 
Another important center based objective function is the \ksum problem where the goal is to minimize the {\em sum} of the radii of the $k$ clusters. Besides being an interesting problem in its own right, the \ksum problem is sometimes preferred over the \kcen problem because it avoids the so-called {\em dissection effect}: an optimal \kcen solution may  place several pairs of close points in two different clusters (\cite{HansenJ97, monmasuri}).  
The problem is known to be $\mathsf{NP}$-hard even in metrics defined by weighted planar graphs and in metrics of constant doubling dimension~\cite{gkkpv10}. Gibson et al.~\cite{gkkpv10} gave an exact randomised algorithm with running time $n^{O(\log{n}\cdot \log{\Delta})}$, where $\Delta$ denotes the aspect ratio (i.e., max. to min. interpoint distance). They also gave a randomised $(1+\veps)$-approximation algorithm with running time $n^{O(\log{n} \cdot \log{\frac{n}{\veps}})}$. Interestingly, there is a polynomial time algorithm~\cite{gibson12} for a special case of points on a plane, where the problem is to cover points using disks.

The \ksum problem was first considered from the approximation algorithms perspective by Doddi et al.~\cite{DoddiMRTW00}. They gave a bi-criteria logarithmic approximation algorithm with $O(k)$ clusters  for \ksum. 
Charikar and Panigrahy~\cite{CharikarP04} improved this result to give the first constant factor 3.504-approximation algorithm for this problem. 
Using coreset based techniques of B\u{a}doiu et al.~\cite{bhi02}, one can obtain a $(1+\veps)$-approximation $2^{O(\frac{k\log{k}}{\veps^2})}\cdot dn^{O(1)}$-time algorithm in the Euclidean setting. Note that all the works mentioned above have been on the uncapacitated version of the problem where there is no upper limit on the number of points that can be assigned to a center.

In this paper, we consider the \capsum problem, where each data point $p$ also specifies a capacity $U_p$ -- this is an upper bound on the size of a cluster centered at $p$. It is worth noting that we only consider the so-called {\em hard} capacitated setting: multiple  centers cannot be opened at the same point. Capacitated versions of clustering problems have been well-studied in the literature and arise naturally in practical settings. Although constant factor approximation algorithm is known for the \capcen problem, obtaining a similar result for the \capmed remains a challenging open problem. 

Motivated by the above discussion, we consider the \capsum problem in the FPT setting. We seek algorithms which have constant approximation ratio but their running time can be of the form $f(k) \poly(n)$, where $f(k)$ can be an exponential function of $k$ (and hence, such algorithms are FPT with respect to the parameter $k$).  Inamdar and Vardarajan~\cite{IV20} had considered this problem in the special case of uniform capacities, i.e., each point  in the input has the same capacity. They gave a 28-approximation algorithm with running time $2^{O(k^2)} \poly(n)$. This was subsequently improved by Bandyapadhyay et al.~\cite{BandyapadhyayL023a} who gave an FPT $(4+\varepsilon)$-approximation algorithm for the uniform capacity setting. The running time of their algorithm  crucially uses the fact that objective function is the sum, i.e., the $L_1$-norm, of the radii of the clusters. Our first result improves this to an FPT $(3+\varepsilon)$-approximation algorithm. Further, our result extends to a constant factor approximation ratio for any $L_p$ norm, where $p \geq 1$. 

Bandyapadhyay et al.~\cite{BandyapadhyayL023a}
also considered the general (i.e., non-uniform capacities) \capsum problem  and gave an FPT $(15+\varepsilon)$-approximation algorithm for this problem. Their argument also extends to general $L_p$ norm. However, it turns out that 
one of the arguments in their proof is incomplete and requires a deeper argument. Motivated by our insights, we refine and extend their proof which not only corrects this step but also yields a better FPT $(9+\varepsilon)$-approximation algorithm. A careful parameter balancing improves this to about 7.6-approximation. Our result also extends to any $L_p$ norm of cluster radii, where $p \geq 1$. 


Our final result shows that there exists a constant $C > 1$ such that, assuming \ETH, any $C$-approximation algorithm for \capsum (in the non-uniform capacity setting) must have exponential running time. This also rules out polynomial time approximation algorithms with approximation ratio better than $C$. 

\subsection{Our Results and Techniques}
In this section, we give a more formal description of our results. Our first result is concerned with the \capsum problem with uniform capacities. The first few steps in our algorithm are analogous to those of~\cite{BandyapadhyayL023a}. Suppose we can guess the radii of the clusters in an optimal solution, say $r_1^\star, \ldots, r_k^\star$. Let $C_i^\star$ denote the set of points in an optimal cluster. We call a cluster large if it has at least $U/k^3$ points (where $U$ denotes the uniform capacity at each point). Again, we can guess the radii of the optimal large clusters in FPT time. Assume that there are $k_L$ large clusters with radii $r_1^\star, \ldots, r_{k_L}^\star$. It is not hard to show that we can sample one point $p_i$ from each large cluster with reasonably high probability. Thus, for each large cluster of radius say $r_i^\star$, we can find a ball $B_i$ of radius $2r_i^\star$ containing it -- the center of such a ball would be the sampled point $p_i$ as above. Now suppose the union of all of the balls $B_i, i \in [k_L^\star]$, covers the entire set of points (although this is a special scenario, it captures the non-triviality of the algorithm).  At this moment, we cannot assign all the points to the chosen centers $p_i$ because of the capacity constraint. The approach of ~\cite{BandyapadhyayL023a} is the following: consider  a graph where the vertex set corresponds to the small optimal balls (i.e., $C_i^\star, i > k_L^\star$) and
the balls $B_i, i \in [k_L^\star]$ constructed by our algorithm. We have an edge between two vertices here if the corresponding sets intersect. 
Since this is a graph on $k$ vertices, we can guess this graph. Now for each connected component in this graph, we can have a single ball whose radius is the sum of the radii of the participating balls in this graph. It is not hard to show that such large radii balls can cover all the small balls and hence can account for the capacity deficit due to the balls $B_i, i \in [k_L^\star]$. Thus one incurs an additional 2-factor loss in the approximation ratio here, resulting in a 4-approximation algorithm. Further, the fact that the algorithm chooses a radius which is the sum of several optimal radii shows that this approach does not extend to $L_p$ norm of radii. 

Our approach is the following: the union of the balls $B_i$ selected by our algorithm covers all the points. We can show that a small extension of these balls, say $B_i'$ for each $B_i$, has the property that they can cover all the points (because even the original balls $B_i$ has this property) and each such ball is assigned only slightly more than $U$ points. Finally, we use a subtle matching based argument to show that the slightly extra points in each of these extended balls $B_i'$ can be covered by a ball whose radius matches the radius of a unique small optimal ball. This allows us to get a better 3-approximation algorithm, and our approach extends to $L_p$ norm of cluster radii as well. Thus, we get the following result:

\begin{theorem}
\label{thm:uniform} There is a randomized $(3+\varepsilon)$-approximation algorithm for \capsum with uniform capacities, where $\varepsilon > 0$ is any positive constant. For $p \geq 1$, there is a randomized $(1+\varepsilon)(2^{2p-1}+1)^{1/p}$-approximation algorithm when the objective function is the $L_p$ norm of the cluster radii. 
The expected running time of both the algorithms is $2^{O(k \log (k/\varepsilon))} \cdot \poly(n)$. 
\end{theorem}

We now consider the general \capsum problem with non-uniform capacities. Again, our initial steps are similar to those in~\cite{BandyapadhyayL023a} -- we initially guess the optimal cluster radii $r_1^\star, \ldots, r_k^\star$; and greedily find a set of balls $B_1, \ldots, B_\ell$, where $\ell \leq k$ and these balls cover the entire set of points. 
Further, each of the balls $B_i$ is supposed to also contain the corresponding optimal cluster $C_i^\star$. At this moment, we need to find some more balls which can be used to satisfy the capacity constraints, i.e., balls which are approximations to $C_{\ell+1}^\star, \ldots, C_k^\star$ respectively. Bandyopadhyay et al.~\cite{BandyapadhyayL023a} approached this problem as follows: for an index $i \in [k] \setminus [\ell]$, one can guess the the set of balls $B_j, j \in [\ell]$ that intersect $C^\star_i$. If any of these balls $B_j$ have smaller radius than $r_i^\star$, we can use a enlarged version $B_j'$ of $B_j$ to cover $C^\star_i$ and treat this as the ball $B_i$. 
The radius of $B_j'$ can be charged to $r^\star_i$. 
Therefore, the non-trivial case arises when $r_i^\star$ is much less than the radii of each of the balls $B_j$ intersecting $C_i^\star$. In such a case, they realized that $C_i^\star$ is contained in the intersection of the balls $B_j'$ defined above. Therefore, they maintain a set of {\em replacement} balls, which are mutually disjoint, and are supposed to contain at least the same number of points as the corresponding cluster $C_i^\star$. 

Ensuring that one can maintain such a collection of disjoint replacement balls that are also disjoint with some desired subset of optimal balls turns out to be non-trivial. In fact, we feel that the proof of Lemma~14 in~\cite{BandyapadhyayL023a},  without a non-trivial fix, has a technical flaw: while inserting a new replacement ball in the set $\B_2$, it may happen that the replacement intersects with an optimal
ball with index in $\I_2$.
 \footnote{ In the proof of Lemma 14 in~\cite{BandyapadhyayL023a} (Lemma 2.5 in the arXiv version), the last line of the first paragraph reads ``If $B_x$ does not intersect  any ball $B_{i'}^\star$ with $i' \in [k] \setminus (I_1 \cup I_2)$,  then by the above discussion we have that setting $B_i = B_x$, $(I_1, I_2 \cup \{i\}, \B_1, \B_2 \cup \{B_i\})$  is a valid configuration.'' However, by doing so, one cannot rule out  the intersection of $B_i$ with a ball $B_{i'}^\star$ with $i' \in I_2$, contradicting the fourth invariant property of a valid configuration: 
``For every $i \in I_2$ and $s \notin I_1$,  $B_s^\star$ and $B_i$ do not intersect.'' It is tricky to note why this is the case. 
They indeed rule out the intersection of $B_i^\star$ with any ball in $\B_2$, but they do not ensure that $B_i$ does not intersect an optimal ball $B_s^\star$  corresponding to the index  $s \in I_2$.
}
Our algorithm uses a more fine grained approach: we first identify a set $P_j$ of points which could contain the desired replacement ball (this step is also there in~\cite{BandyapadhyayL023a}). Now we guess the set $Z$ of optimal clusters which could intersect this  candidate set $P_j$. The algorithm then runs over several iterations, where in each iteration it shrinks the candidate set of points $P_j$ containing the replacement ball, and hence the list $Z$ (see~\Cref{algo:update}). This process terminates when it either identifies a replacement ball or a ball containing a cluster $C_r^\star$. To get a better approximation ratio, we also maintain several set of balls with varying properties. We prove the following result:

\begin{theorem}
\label{thm:nonuni} There is a randomized $(4 + \sqrt{13} + \varepsilon)$-approximation algorithm for \capsum, where $\varepsilon > 0$ is any positive constant. For $p \geq 1$, there is a randomized $(4 + \sqrt{13} + \varepsilon)$-approximation algorithm when the objective function is the $L_p$ norm of the cluster radii. 
The expected running time of both the algorithms is $2^{O(k^3 + k\log(k/\varepsilon))} \cdot \poly(n)$. 
\end{theorem}

Our final result shows that the exponential dependence on running time is necessary if we want a $c$-approximation algorithm, where $c$ is a constant larger than 1. This result assumes that the \eth holds and uses a reduction from \vc on bounded degree graphs. A nearly exponential time lower bound for approximating \vc on such graphs follows from \cite{dinur}. More formally, we have the following result:

\begin{theorem}
    \label{thm:eth}
    Assume that \ETH holds. 
    Then there exist positive constants $c_1, c_2$, where $c_1 > 1,$ such that
    any algorithm for \capsum with running time at most $2^{c_2 k/\polylog k} \cdot \poly(n)$ must have approximation ratio at least $c_1$. 
\end{theorem}


\subsection{Preliminaries}
We define the $\capsum$ problem. An instance $\cI$ of this problem is specified by a set of $n$ points $P$ in a metric space, a parameter $k$ and capacities $U_p$ for each point $p \in P$. A solution to such an instance needs to find a subset $C \subseteq P$, denoted ``centers'',  such that $|C|=k$. Further, the solution assigns each point in $P$ to a unique center in $C$ such that for any $c \in C$, the total number of points assigned to it is at most its capacity $U_c$. It is worth noting that we are interested in the so-called ``hard'' capacity version, where one cannot locate multiple centers at the same point. 

Assuming $C = \{c_1, \ldots, c_k\}$, let $C_j$ denote the set of points assigned to the center $c_j$ by this solution. We refer to $C_j$ as the {\em cluster} corresponding to $C_j$ and define its radius $r_j$ as $\max_{p \in P_j} d(p, c_j)$. The goal is to find a solution for which the sum of cluster radii is minimized, i.e., we wish to minimize $\sum_{j=1}^k r_j$. In the $L_p$-norm version of this problem, the specification of an instance remains as above, but the objective function is given by $\left( \sum_{j=1}^k r_j^p \right)^{1/p}$. 
Given an instance $\cI$ as above, let $\opt(\cI)$ denote an optimal solution to $\cI$. Let $r^\star_1, \ldots, r^\star_k$ be the radii of the $k$ clusters in this optimal solution. By a standard argument,  we can show that one can guess close approximations to these radii. 
\begin{lemma}\cite{BandyapadhyayL023a}
    \label{cl:radii}
    Given a positive constant $\varepsilon > 0$, there is an $O(2^{O(k \log (k/\varepsilon))} \cdot n^3)$ time algorithm that outputs a list $\cal L$, where each element in the list is a sequence $(r_1, \ldots, r_k)$ of non-negative reals, such that the following property is satisfied: there is a sequence $(r_1, \ldots, r_k) \in {\cal L}$ such that for all $j \in [k]$, $r_j \geq r_j^\star$ and $\sum_{j=1}^k r_j \leq (1+\varepsilon) \sum_{j=1}^k r_j^\star.$ Further, the size of the list $\cal L$ is $O(2^{O(k \log (k/\varepsilon))} \cdot n^2)$.
\end{lemma}

Given a point $x$ and non-negative value $r$, let $B(x,r)$ denote the {\em ball} of radius $r$ around $x$, i.e., $B(x,r) := \{p \in P: d(p,x) \leq r\}$, where $d(a,b)$ denote the distance between points $a$ and $b$. For a positive integer $r$, we shall use $[r]$ to denote the set $\{1, \ldots, r\}$.

\subsection{Related Work}
In this section, we mention a few related research works on the sum-of-radii problem that were not discussed earlier in the introduction. Note that most of the work has been for the uncapacitated case.
Behsaz and Salavatipour~\cite{bs15} give an exact algorithm for the sum of radii problem in metrics induced by unweighted graphs for the case when singleton clusters are disallowed.
Friggstad and Jamshidian~\cite{fm22} give a $3.389$-approximation algorithm for the sum of radii problem in general metrics. 
This is an improvement over the $ 3.504$ approximation of Charikar and Panigrahy~\cite{CharikarP04}, the first constant factor approximation algorithm for this problem.
Bandyapadhyay and Varadarajan~\cite{bv16} discuss the variant of the problem where the objective function is the sum of the $\alpha^{th}$ power of the radii. The algorithm can output $(1+\veps)k$ centers in the variant they study.
Henzinger et al.~\cite{hlm20} give a constant approximation algorithm for metric spaces with bounded doubling dimension in the {\em dynamic setting} where points can appear and disappear.
There also have been works for special metrics such as two-dimensional geometric settings (e.g.,~\cite{gkkpv10}).
The related problem of sum-of-diameters has also been studied (e.g.,~\cite{DoddiMRTW00,bs15}).

We now give an outline of rest of the paper. In~\Cref{sec:uni}, we prove~\Cref{thm:uniform}, i.e., we present an FPT approximation algorithm for \capsum in the uniform capacity case. The more general non-uniform capacity case is considered in~\Cref{sec:nonuni}, where we prove~\Cref{thm:nonuni}. 

%% file: 2-Uniform-capacity.tex
\section{Uniform Capacities}
\label{sec:uni}
In this section, we consider the special case of \capsum when all the capacities $U_p$ are the same, say $U$. We give some notation first. Fix an optimal solution $\cO$ to an instance $\I$ given by a set of $n$ points $P$. Let $c_1^\star, \ldots, c_k^\star$ be the $k$ centers chosen by the optimal solution.  For an index $j \in [k]$, let $C_j^\star$ denote the set of points assigned to $c_j^\star$ by this solution (i.e., the cluster corresponding to the center $c_j^\star$). Let $r_1^\star, \ldots, r_k^\star$ be the radii of the corresponding clusters $C_1^\star, \ldots, C_k^\star$. We also fix a parameter $\gamma := \frac{1}{k^2}$.

\begin{algorithm}[H]
\caption{An iteration of the algorithm for $\capsum$ when all capacities are $U$.}
{\bf Input:} Set $P$ of $n$ points, parameter $k, k_L^\star$, radii $r_1, \ldots, r_k$, capacity $U$. \;
Initialize a set $\B$ to empty. \;
Initialize an index set $I$ to $\{1, \ldots, k_L^\star\}$. \;
\For{$j=1, \ldots, k_L^\star$ \label{l:for0}}{ \label{l:large}
    Choose a point $c_j \in P$ uniformly at random. \;
    Add $B_j := B(c_j, 2r_j)$ to $\B$ \label{l:cj}. 
}
\While{there is a point $x$ not covered by the balls in $\B$ \label{l:xchoose}}{ 
(If $I = [k]$, output {\bf fail}). \label{l:failwhile} \;
Choose an index $j \in [k] \setminus I$ uniformly at random. \label{l:guess1}\;
   Add $B_j := B(x, 2r_j)$ to $\B$ and add $j$ to $I$ and define $c_j := x$. \label{l:ball}
}
Initialize sets $T_j \subset [k], j \in I,$ to emptyset. \;
\For{each $i \in [k] \setminus I$ \label{l:ki}}{
 Choose an index $j \in I$ uniformly at random and add $i$ to $T_j$. \label{l:tj}
}
Define $I' := \{j \in I: |T_j| > 0 \}$. \; \label{l:I'}
\For{each $j \in I$ \label{l:bjb}}{
 \If{$j \in I'$}{
    Define $B_j' := B(x, 2r_j + 2R_j)$ and $U_j' := U (1+ \gamma)$, where $R_j := \max_{i \in T_j} r_i$. \label{l:bj'}
 }
 \Else{
    Define $B_j' := B_j, U_j' = U. $ \label{l:bj}
 }
Find disjoint subsets $G_j \subseteq B_j'$ for each $j \in I$ such that $|G_j| \leq U_j'$ and $P = \cup_j G_j$. \; \label{l:gj}
(terminate with failure if such subsets $G_j$ do not exist) \;
Let $I'' \subseteq I'$ be the index set consisting of indices $j$ such that $|G_j| > U$. \label{l:i''} \;
Call {\bf Redistribute}($\{G_j: j \in I'' \}, \{c_j: j \in I\}, \{r_j : j \in [k_L]\}$). \label{l:dist} \;
(terminate with failure if {\bf Redistribute} outputs {\bf fail}) \;
Let $\{(w_j, A_j): j \in I''\}$ be the clustering returned by {\bf Redistribute}. \;
{\bf Output} $\{(c_j, G_j): j \in I \setminus I''\} \cup \{(c_j, G_j \setminus A_j): j \in I''\} \cup \{(w_j, A_j): j \in I'' \}.$ \label{l:out}
}
\label{algo:uni}
\end{algorithm}

\begin{definition}
    Call the optimal  cluster $C_j^\star$ {\em large} if $|C_j^\star| \geq \frac{\gamma U}{k}$; otherwise call it {\em small}. Let $k_L^\star$ denote the number of large clusters in the optimal solution $\cO$. 
\end{definition}

Assume without loss of generality that the clusters $C^\star_1, \ldots, C^\star_{k_L^\star}$ are large (and the rest are small). Using~\Cref{cl:radii}, we can assume that we know radii $r_1, \ldots, r_k$ satisfying the condition that $r_j \geq r_j^\star$ for all $j \in [k]$ and $\sum_j r_j \leq (1+\varepsilon) \sum_j r_j^\star$. By cycling over the  $k$ possible choices of $k_L^\star$, we can also assume that we know this value. The algorithm is given in~\Cref{algo:uni}. It begins by guessing the center $c_j$ of each large optimal cluster $C_j^\star$ and defines $B_j$ as the ball of radius $2 r_j$ around $c_j$ (line~\ref{l:large}). The intuition is that $c_j$ may not be equal to the center $c_j^\star$ but will lie inside the cluster $C_j^\star$ with reasonable probability, and in this case, the ball $B_j$ will contain $C_j^\star$. Let us assume that this event happens. Now the algorithm adds some more balls to the set $\B$ (that maintains the set of balls constructed so far). Whenever there is a point $x$ that is not covered by the balls in $\B$, we guess the index $j$ of the optimal cluster $C_j^\star$ containing $x$ (line~\ref{l:guess1}). We add a ball of radius $2r_j$ around $x$ to $\B$: again the intuition is that if $x\in C_j^\star$, then this ball contains $C_j^\star$.

The index set $I$ maintains the set of indices $j$ for which we have approximation to $C_j^\star$ in $\B$. At this moment, the set of balls in $\B$ cover the point set $P$ but we are not done yet because we need to assign points to balls while maintaining the capacity constraints. Now, for each index $i \notin I$, we guess the index of a ball $B_j \in \B$ such that $B_j$ intersects $C_i^\star$ (such a ball must exist since the balls in $\B$ cover $P$, and in case there are more than one candidates for $B_j$, we pick one arbitrarily). We add the index $i$ to a set $T_j$ (line~\ref{l:tj}). Now we define $I'$ as the subset of $I$ consisting of those indices $j$ for which $T_j$ is non-empty (line~\ref{l:I'}). Now, for each $j \in I'$, let $R_j$ denote the maximum radius of any of the balls corresponding to $T_j$. We replace $B_j$ by a larger ball $B_j'$ by extending the radius of $B_j$ by $2R_j$ -- this ensures that $B_j'$ contains $C_i^\star$ for any $i^\star \in T_j$. 
Further we allow $B_j'$ to have $U_j' := U(1+\gamma)$ points assigned to it (line~\ref{l:bj'}). For indices $j \in I \setminus I'$, we retain $B_j', U_j'$ as $B_j, U_j$ respectively (line~\ref{l:bj}). Now, we find a subset $G_j$ of each ball $B_j'$ such  that $|G_j| \leq U_j'$ and $\cup_j G_j$ covers all the points $P$ (line~\ref{l:gj}).
We shall show such a subset exists if all our guesses above our correct. Otherwise we may not be able to find such sets $G_j$; and in this case the iteration terminates with failure.  The intuition is that  one feasible choice of $G_j$ is as follows: for each $j \in I \setminus I'$, $G_j := C_j^\star$, whereas for an index $j \in I'$, we set $G_j$ to be the union of $C_j^\star$ and all the clusters $C_i^\star$, where $i \in T_j$. 
Since each small cluster has at most $\gamma U/k$ points, the total number of points in the latter clusters is at most $\gamma U$. 
But we had set $U_j'$ to $(1+\gamma)U$. It is also easy to check that we can find such sets $G_j$ by a simple $b$-matching  formulation.
Now we consider those subsets $G_j$ for which $|G_j| > U$ (this can only happen if $j \in I'$) and let $I''$ denote the corresponding index set (line~\ref{l:i''}). 
Finally, we call the procedure ${\bf Redistribute}$ in line~\ref{l:dist} to redistribute the points in $G_j, j \in I''$ such that each such set has at most $U$ points. Note that so far we have only constructed $|I|$ balls in $\B$ and we can still add $k-|I|$ extra balls. The procedure ${\bf Redistribute}$ returns these extra balls -- we finally return these balls and remove suitable points from $G_j, j \in I''$ (line~\ref{l:gj}). We shall show later that $|I''|$ is less than $k_L$.  Again, note that if our random choices were bad, it is possible that {\bf Redistribute} returns failure, in which case the iteration ends with failure.  

\begin{algorithm}[H]
\caption{Algorithm {\bf Redistribute}}
{\bf Input:} $(G_1, \ldots, G_h)$, where $G_i$'s are pair-wise disjoint subsets of $P$, where $|G_i| > U$ for all $i$;  a subset $C$ of $P$, and a set $R = \{r_1, \ldots, r_{k_L}\}$ of radii. \;
\For{each ordered subset $(r_{\sigma_1}, \ldots, r_{\sigma_h})$ of $R$ \label{l:forguess}}{ \label{l:for}
    Construct a bipartite graph $H=(V_L, V_R, E)$ as follows. \;
    $V_L$ has one vertex $v_i$ for each $G_i, i \in [h]$. $V_R$ is defined as $P \setminus C$. \label{l:graph}\;
    Add  an edge $(v_i, w), v_i \in V_L, w \in V_R,$ iff $|B(w, r_{\sigma_i}) \cap G_i| \geq \gamma U$. \label{l:edge} \;
    \If{$H$ has a matching that matches all vertices in $V_L$  \label{l:matching}}{
       Suppose $v_i$ is matched with $w_i \in R$ for each $i \in [h]$. \label{l:match1} \;
       Let $A_i$ be a subset of $B(w_i, r_{\sigma_i}) \cap G_i$ such that $|A_i| = \gamma U$. \label{l:match2} \;
       Return $\{(w_i, A_i): i \in [h]\}$ (and end the procedure).  
    
    }

}
Return {\bf fail}. 
\label{algo:redist}
\end{algorithm}

We now describe the algorithm {\bf Redistribute} in~\Cref{algo:redist}. 
The procedure receives three  parameters: the first parameter is a class of subsets $\{G_1, \ldots, G_h\}$ which are mutually disjoint.
These correspond to the sets $G_i, i \in I''$ constructed in~\Cref{algo:uni} (and hence $h = |I''|). $ 
The second parameter is a set $C$, which corresponds to the set of centers chosen by~\Cref{algo:uni} (till the call to this procedure) and the third parameter is a set of radii $\{r_1, \ldots, r_{k_L}\}$ which correspond to the radii of the large clusters. 
We shall ensure that $k_L \geq h$. Recall that the goal is to identify balls which can take away $\gamma U$ points from each of the sets $G_i$. The radii of these balls shall come from the set $R:= \{r_1, \ldots, r_{k_L}\}$, and we shall show that we can associate a unique radius with each of the desired balls. Thus, the procedure tries out all ordered subsets of size $h$ of $R$ (line~\ref{l:forguess}). For each such ordered subset $(r_{\sigma_1}, \ldots, r_{\sigma_h})$, we try to remove $\gamma U$ points from each subset $G_i$ by using a suitable ball of radius $r_{\sigma_i}$. Thus, we construct a bipartite matching instance as follows: the left side has one vertex for each subset $G_i$ and the right side has one vertex for each potential center of a ball (line~\ref{l:graph}). Now, we add an edge between a vertex $v_i$ corresponding to $G_i$ on left and a vertex $w$ on the right side if $G_i \cap B(w, r_{\sigma_i})$ has size at least $\gamma U$ (line~\ref{l:edge}). If this graph has a perfect matching, then we identify the desired subsets $A_i \subseteq G_i$ of size $\gamma U$ (line~\ref{l:match2})   and return these. 

\subsection{Analysis}
In this section, we prove correctness of the algorithm. We shall show that with non-zero probability the algorithm outputs a 3-approximate solution. 
We first define the set of desirable events during the random choices made by~\Cref{algo:uni}:
\begin{itemize}
    \item $\cE_1$: For each $ j \in [k_L^\star]$, the point $c_j$ chosen in line~\ref{l:cj} of~\Cref{algo:uni} belongs to the cluster $C_j^\star$. 
    \item $\cE_2$: For each point $x$ chosen in line~\ref{l:xchoose} of \Cref{algo:uni}, the index $j$ chosen in line~\ref{l:guess1} satisfies the property that $x \in B_j^\star$.  Further, the algorithm does not output {\bf fail} in line~\ref{l:failwhile}.
    \item $\cE_3$: For each index $i$ considered in line~\ref{l:ki} in \Cref{algo:uni}, the index $j \in I$ selected in line~\ref{l:tj} satisfies the property  that $B_j \cap B_i^\star$ is non-empty. 
    \item $\cE_4$: The set of centers $\{c_j: j \in I\}$ selected by~\Cref{algo:uni} is disjoint from the set of optimal centers of the large balls, i.e., $\{c_j^\star: j \in [k_L^\star]\}.$
\end{itemize}

We first show that all of these desirable events happen with non-zero probability. 
\begin{lemma}
    \label{lem:event}
    Assuming $n \geq 2k^5$, all the events $\cE_1, \cE_2, \cE_3, \cE_4$ together happen with probability at least $\frac{1}{k^{O(k)}}.$ 
\end{lemma}
\begin{proof}
    We first check event $\cE_4 \cap \cE_1$. Consider an iteration $j$ of the {\bf for} loop in line~\ref{l:for0}. Conditioned on the choices in the previous $j-1$ iterations, the probability that $c_j \in C_j^\star \setminus \{c_j^\star: j \in [k_L^\star]\}$ is at least $\frac{|C_j^\star|-k}{n} $. Since $C_j^\star$ is large, we know that $|C_j^\star| \geq \gamma U/k \geq n/k^4$ (since $\gamma = 1/k^2$ and $U \geq n/k$). Using the fact that $n \geq 2k^5$, we get 
    $$\Pr[\cE_1 \cap \cE_4] = \Pr[c_j \in C_j^\star \setminus \{c_j^\star: j \in [k_L^\star]\}, \text{for all $j=1, \ldots, k_L^\star$} ] \geq \left(\frac{1}{2 k^4}\right)^k = \frac{1}{2^k k^{4k}}$$ 

    Now we consider $\cE_2$. We condition on the coin tosses before the {\bf while} loop in line~\ref{l:xchoose} such that $\cE_1 \cap \cE_4$ occur. This implies that for every $j \in [k_L^\star]$, $C_j^\star \subseteq B_j$. 
    
    The probability that we correctly guess the index $j$ such that  the cluster $C_j^\star$ contains $x$ (in line~\ref{l:xchoose}) is $1/k$. Since there can be at most $k$ iterations of the {\bf while} loop, the probability that this guess is correct for each point $x$ chosen in line~\ref{l:xchoose} is at least $1/k^k$. Further, if this guess is always correct, then $B_j$ contains $C_j^\star$ for all $j \in I$. This shows that if $I$ becomes equal to $[k]$, then the balls $B_j$ would cover $P$ and hence, we won't output {\bf fail}. Thus, we see that 
    $$\Pr[\cE_2 | \cE_1 \cap \cE_4] \geq 1/k^k.$$

    Finally we consider $\cE_3$. Again condition on the events before the {\bf for} loop in line~\ref{l:ki} and assume that $\cE_1 \cap \cE_2 \cap \cE_4$ occur. There are at most $k$ iterations of the ${\bf for}$ loop. Since $\cup_{j \in I} B_j = P$, there must exist an index $j \in I$ such that $B_j$ intersects $C_i^\star$. Therefore, the probability that we guess such an index $j$ in line~\ref{l:tj}
    is at least $1/k$. Thus, we get 
    $$ \Pr[\cE_3| \cE_1 \cap \cE_2 \cap \cE_4] \geq 1/k^k.$$
    Combining the above inequalities, we see that $\Pr[\cE_1 \cap \cE_2 \cap \cE_3 \cap \cE_4] \geq \frac{1}{2^k k^{6k}}.$ 
    This implies the desired result. 
\end{proof}
We now show that the sets $G_j$ as required in line~\ref{l:gj} exist and can be found efficiently. 
\begin{claim}
 Assume that the events $\cE_1, \ldots, \cE_4$ occur. Let $B_j', U_j'$ be as defined in lines~\ref{l:bjb}--\ref{l:bj} of \Cref{algo:uni}. Then there exist mutually disjoint subsets $G_j \subseteq B_j'$ for each $j \in I$ such that $P= \cup_{j \in I} G_j$ and $|G_j| \leq U_j'$ for each $j \in I$. Further, the subsets $G_j$ can be found in $\poly(n)$ time. 
\end{claim}
\begin{proof}
    Since events $\cE_1, \cE_2, \cE_3$ occur, it is easy to check that for each $j \in I$, $B_j'$ contains $C_j^\star \cup \bigcup_{h \in T_j} C_h^\star$. Further, if $T_j$ is non-empty, then 
    $$|C_j^\star| + \sum_{h \in T_j} |C_h^\star| \leq U(1+\gamma),$$
    because each of the clusters $C_h^\star$ is small. Thus, one feasible choice for the subsets $G_j$ is as follows: for each $j \in I$ such that $T_j$ is empty, define $G_j = C_j^\star$, otherwise define $G_j = C_j^\star \cup \bigcup_{h \in T_j} C_h^\star$. This proves the existence of the desired subsets $G_j$. It is easy to check that such a collection of subsets can be found by standard flow-based techniques, and hence, would take $\poly(n)$ time. 
\end{proof}
We now show that the {\bf Redistribute} outputs the desired subsets.
\begin{lemma}
 Assume that the events $\cE_1, \ldots, \cE_4$ occur. Then~\Cref{algo:redist} does not fail. 
\end{lemma}
\begin{proof}
    We first consider the following bipartite graph $H'=(V_L', V_R', E')$ (this is for the purpose of analysis only): the vertex set $V_L'$ has one vertex $v_i$ for each of subsets $G_i$, and hence, is same as the set $V_L$ considered in line~\ref{l:graph} of~\Cref{algo:redist}. $V_R'$ has a vertex $w_j$ for each optimal cluster $C_j^\star, j \in k_L^\star$.  We add an edge $(v_i, w_j)$ iff $|G_i \cap C_j^\star| \geq \gamma U$. We claim that there is a matching in this graph that matches all the vertices in $V_L'$. Suppose not. Then there is a subset $X$  of $V_L'$ such that $|N(X)| < |X|$, where $N(X)$ denotes the neighborhood of $X$. Now, 
    \begin{align*}
        \sum_{i \in X} |G_i| & = \sum_{i \in X} \sum_{j \in [k]} |C_j^\star \cap G_i| = \sum_{j \in [k]} \sum_{i \in X} |C_j^\star \cap G_i| \\
        & \leq \sum_{j \in N(X)} |C_j^\star| + \sum_{j \in [k] \setminus N(X)} \sum_{i \in X} \gamma U \\
        & \leq |N(X)| U + k^2 \gamma U \\
        & \leq U(|X|-1) + U = U|X|,
    \end{align*}
    where the second inequality uses the fact that if $i \in [k] \setminus N(X)$, then $|C_j^\star \cap G_i| \leq \gamma U$. Indeed, if $C_j^\star$ is small, this follows from the observation that $|C_j^\star| \leq \gamma U$. Otherwise, $j \in [k_L^\star]$ and hence we have a vertex $w_j$ corresponding to $C_j^\star$ in $H'$. Since $(v_i, w_j)$ is not an edge in $H'$, it follows that $C_j^\star \cap G_i$ has size at most $\gamma U$. Now we get a contradiction because for each $i$, $|G_i| > U$ and hence, $\sum_{i \in X} |G_i| > U|X|$ 
    
    Thus we have shown that $H'$ has a matching that matches all the vertices in $V_L'$ -- let $w_{\sigma_i}$ be the vertex in $V_R'$ which is matched to $v_i \in V_L'$ by this matching. Now consider the iteration of the {\bf for loop} in line~\ref{l:forguess} in~\Cref{algo:redist} where the sequence of radii is given by $(r_{\sigma_1}, \ldots, r_{\sigma_h})$. We claim that the graph $H$ constructed in this iteration (in line~\ref{l:graph}) has a matching that matches all the vertices in $V_L$. Indeed, we can match the vertex $v_i \in L$ with $c_{\sigma_i}^\star$ because $\gamma U \leq |G_i \cap C_{\sigma_i}^\star| \leq |G_i \cap B(c_{\sigma_i}^\star, r_{\sigma_i})| $ -- we use the fact that the event $\cE_4$ occurred, and hence, $c_{\sigma_i}^\star \in C$. 

    Now, let $M$ be the matching found in line~\ref{l:matching} and suppose $v_i$ is matched with a vertex $w_i \in R$. By definition of $H$, $B(w_i, r_{\sigma_i}) \cap G_i$ has size at least $\gamma U$. Thus, we can find the desired subset $A_i$ in line~\ref{l:match2}. Note that the subsets $A_i$ are mutually disjoint since the sets $G_i$ are mutually disjoint. Further the points $w_i, i \in H,$ are also distinct since $M$ is a matching. Thus, the procedure ${\bf Redistribute}$ returns $h$ clusters, each containing $\gamma U$ points.  
\end{proof}

It follows from the results above that the set of $k$ clusters returned by~\Cref{algo:uni} in line~\ref{l:out} cover all the points in $P$ and satisfy the capacity constraints. We now consider the objective function value of this solution. 

\begin{lemma}
    \label{lem:cost_uni}
    Assume that the events $\cE_1, \ldots, \cE_4$ occur. Then the total sum of the radii of the clusters returned by~\Cref{algo:uni} is at most $3 \sum_{j \in [k]} r_j$. Further, for any $p \geq 1$, the total $L_p$ norm of the radii of the clusters returned by this algorithm is at most $(2^{2p-1} + 1)^{1/p} \left( \sum_{j \in [k]} r_j^p \right)^{1/p}.$
\end{lemma}
\begin{proof}
    For each $j \in I$, the ball $B_j$ constructed during lines~\ref{l:large}--\ref{l:ball} has radius at most $2r_j$. Now, the ball $B_j'$ constructed during the {\bf for} loop in line~\ref{l:bjb} has radius at $2r_j + \sum_{h \in T_j} 2r_h$. Therefore, the total radii of the balls in $B_j'$ is at most 
    $$ \sum_{j \in I} \left( 2r_j + \sum_{h \in T_j} 2r_h \right) \leq 2 \sum_{j \in [k]} r_j,$$
    where the inequality follows from the fact that sets $\{j\} \cup T_j, $ where $j \in I$, are mutually disjoint. Finally, the total sum of the radii of the sets $A_i$ returned by {\bf Redistribute} is at most $\sum_{j \in [k]} r_j$. This proves the first statement in the Lemma. 

    Further, the $L_p$-norm of the radii of the clusters output by this algorithm is at most 
    $$ \sum_{j \in I} (2r_j + 2R_j)^p + \sum_{j \in [k]}  r_j^p \leq \sum_{j \in [k]} 2^{2p-1} r_j^p + \sum_{j \in [k]} r_j^p, $$
    where $R_j$ is as defined in line~\ref{l:bj'} of~\Cref{algo:uni}. This completes the proof of the desired result. 
\end{proof}

Thus,  with probability at least $1/k^{O(k)}$ (\Cref{lem:event}),~\Cref{algo:uni} outputs a feasible solution with approximation guarantees as given by~\Cref{lem:cost_uni}. Repeating~\Cref{algo:uni} $k^{O(k)}$ times and using~\Cref{cl:radii} yields~\Cref{thm:uniform}.

%% file: 3-Non-uniform-capacity.tex
\section{Non uniform capacities}
\label{sec:nonuni}
In this section, we consider the general case of $\capsum$ when points can have varying capacities. Recall that for a point $p$, $U_p$ denotes the capacity of $p$. We first give an informal description of the algorithm. Let $r_1^\star, \ldots, r_k^\star$ denote the radii of the $k$ optimal clusters $C_1^\star, \ldots, C_k^\star$. Let $c_1^\star, \ldots, c_k^\star$ be the centers of these clusters respectively. Let $B_j^\star$ denote the ball $B(c_j^\star, r_j^\star)$ -- note that $C_j^\star \subseteq B_j^\star$. 
As in the case of uniform capacities (\Cref{sec:uni}), we begin by assuming that we know radii $r_1, \ldots, r_k$ satisfying the conditions of~\Cref{cl:radii}, i.e.,  $r_j \geq r_j^\star$ and $\sum_{j=1}^k r_j \leq (1+\varepsilon) \sum_{j=1}^k r_j^\star$, where $\varepsilon > 0$ is an arbitrarily small constant. 

Our algorithm maintains three subsets of balls, namely $\B_1, \B_2, \B_3$. Each ball in these sets corresponds (in a sense that we shall make clear later) to a unique optimal ball $B_j^\star$. Thus we denote balls in $\B_1 \cup \B_2 \cup \B_3$ as $B_j$, where the index $j$ denotes the fact that $B_j$ corresponds to $B_j^\star$. Further, we  use $c_j$ to denote the center of $B_j$.   We maintain three sets of (mutually disjoint) indices $I_1, I_2, I_3 \subseteq [k]$ respectively. We also maintain an index set  $I_4 := [k] \setminus (I_1 \cup I_2 \cup I_3)$ denoting those indices for which we have not assigned a ball yet. 

The set $\B_1$ is obtained by the following greedy procedure: while there are points not covered by the balls in $\B_1$, we pick an uncovered  point $p$, guess the index of the optimal cluster covering it, say $j$, and add the ball $B(c_j, 3r_j)$ to $\B_1$, where $c_j$ is a high capacity point close to $p$ (the actual process is slightly nuanced, but the details will be clarified in the actual algorithm description).  When this process ends, we have  a set of balls covering $P$ such that each ball in $\B_1$ covers a corresponding (unique) optimal ball. The set $\B_1$ remains unchanged during rest of the algorithm. We now need to add more balls to this solution in order to satisfy the capacity constraints. 

Such balls shall be added to the sets $\B_2$ and $\B_3$ respectively (and the index sets $I_2$ and $I_3$ shall maintain the correspondence with unique clusters in the optimal solution respectively). Roughly, a ball $B_j \in \B_2$ of the form $B(c_j,r)$ shall satisfy the condition that $r \leq 5r_j$ and $B^\star_j \subseteq B_j$. Our algorithm shall never remove a ball that once gets added to $\B_2$. A typical setting when we can add a ball to $\B_2$ is the following: Suppose there is an index $j \in I_4$ such that the optimal ball $B_j^\star$ intersects a ball $B_h  \in \B_1 \cup \B_2$ and $r_j$ is at least the radius of $B_h$ (we do not know this fact a priori, but our algorithm can {\em guess} such cases). In this case, we find a high capacity point $x$ in the vicinity of $c_h$, such that $B_j = B(x, 5r_j)$ covers $B_j^\star$.


A ball $B_j$ added to the set $\B_3$ shall have the property that it does not intersect any of the balls $B^\star_h, h \in I_3 \cup I_4$ (again, we cannot be certain here since we do not know the optimal balls, but we shall show that in one of the cases guessed by our algorithm, this invariant will be satisfied). A ball $B_j$ once added to $\B_3$ can get deleted, but in this case we will add a corresponding ball to $\B_2$, i.e., we shall remove the index $j$ from $I_3$ and add it to $I_2$. This can happen because of the following reason: suppose we have identified a ball $B_h, h \in I_4$ (centered around a point $c_h$), that we want to add to $\B_3$. But it intersects the optimal ball $B_j^\star$ for an index $j \in I_3$ (again, we {\em guess} this fact). Now if $r_j \geq r_h$, it follows that, as before, we can find a high capacity point $x$ in the vicinity of $c_h$ such that $B_j = B(x, 5 r_j)$ covers $B_j^{\star}$. 

Identifying a ball which gets added to $\B_3$ is at the heart of our technical contribution. Essentially we start with a suitable ball $B$ (picked from a certain set of possibilities) which we would like to add to $\B_3$. As long as there is an optimal ball $B_j^\star, j \in I_3 \cup I_4$ intersecting $B$, we able transfer an index from $I_3$ to $I_2$ or suitably shrink the possibilities for identifying the ball $B$. 

We mention one final technicality. For each ball $B_j \in \B_3$ maintained by the algorithm, we shall maintain a subset $C_j$. The subset $C_j$ is meant to capture the subset of points that will be actually be part of the $j^{th}$ cluster output by the algorithm (although the final output may be a subset of $C_j$).  
The reason for this is as follows.  We would like to maintain the following invariant during the algorithm:

\begin{itemize}
    \item[] {\bf Invariant 1:} Let $j, j' \in I_3$ be two distinct indices. Then $C_j \cap C_{j'} = \emptyset$. Further, for any $j \in I_3$ and $i \in I_3 \cup I_4, $ $B_i^\star \cap C_j = \emptyset.$
\end{itemize}

\subsection{Algorithm Description}
We now give a formal description of our algorithm for $\capsum$ in~\Cref{algo:nonuni}. Here we describe  one iteration of the algorithm and shall show that it outputs the desired solution with probability at least $\frac{1}{2^{O(k^3)}}.$
As mentioned above, there are three sets of balls $\B_1, \B_2, \B_3$ and corresponding (mutually disjoint) index sets $I_1, I_2, I_3$. The set $I_4$ which is initialized to $[k]$ stores the index sets for which we haven't added a ball in $\B_1 \cup \B_2 \cup \B_3$. 

We first explain a subroutine that we shall be using repeatedly during the algorithm: \\
{\bf InsertBall}($p, j, r, u)$). Here $u$ is an index in $\{1,2,3\}$. 
The procedure is supposed to insert the ball $B(p, r)$ centered at $p$ to the set $\B_u$ (and add the index $j$ to $I_u$ -- the index $j$ before this step lies in $I_4$, the set of {\em unused} indices). However, there is a caveat -- we want to avoid the case when $p$ happens to be same as the optimal center $c_h^\star$ for some other $h \neq j$. The reason is that in a subsequent step, the algorithm may try to insert a ball that is an approximation of $B_h^\star$.
At this moment, it may happen that $c_h^\star$ is the only feasible choice for a center because all other points close to $c_h^\star$ have very low capacity. Therefore, the procedure {\bf InsertBall} needs to check if this is the case. In particular, it {\em guesses} this fact, and if indeed $p$ happens to be same as $c_h^\star$, it places a ball of radius $r_h$ around $c_h^\star$ in $\B_1$. The algorithm maintains a global index set $I^\star$ consisting of those indices $j$ for which it has executed this step, i.e., for which the center of the corresponding ball $B_j \in \B_1$ is same as $c_j^\star$. Note that $I^\star$ is a subset of $I_1$. Whenever it adds such an index to $I^\star$, it sets $I_1$ to $I^\star$, $\B_1$ to $\{B(c_i, r_i), i \in I^\star\}$ and resets $I_2, I_3, \B_2, \B_3$ to empty (see line~\ref{l:remove} in~\Cref{algo:insert}), and then restarts the ~\Cref{algo:nonuni} from line~\ref{l:while1}  

For sake of clarity of the algorithm, we shall always assume that the choice among the two options~(i) and (ii) taken by this algorithm is correct, and also that whenever the algorithm chooses option (ii), it chooses $h$, satisfying $p = c_h^\star$. The reason is that we will later see that this procedure is called $O(k^2)$ times, and it chooses the option (ii) at most $k$ times. Hence, the probability that procedure makes the correct choices each time is $1/2^{O(k^2)}$, which suffices for our purpose. 

\begin{algorithm}[H]
\caption{The  procedure {\bf InsertBall}($p, j, r, u)$).}
\label{algo:insert}
{\bf Input:} Candidate center $p$, radius $r$, index $j \in I_4$, index $u \in \{1,2,3\}.$ \;
Perform one of the following two steps with equal probability: \;
(i) Add the ball $B(p, r)$ to $\B_u$ and the index $j$ to $I_u$ (and remove $j$ from $I_4$). Set $c_j = p$. \label{l:step1} \;
(ii) Guess an index $h \in [k] \setminus I^*$ uniformly at random. Set $c_h = p$.  \label{l:step2} \;
\quad \quad set $I_2, I_3, \B_2, \B_3$ to empty. \label{l:remove} \;
\quad \quad Add $h$ to $I^\star$
\;
\quad \quad Set $I_1 = I^\star$ and $\B_1 = \{B(c_i, r_i), i \in I^\star\}$
\;
\quad \quad Go to line~\ref{l:while1} of ~\Cref{algo:insert}
\end{algorithm}

We now describe~\Cref{algo:insert}. 

In lines~\ref{l:while1}--\ref{l:while2}, 
we add balls to $\B_1$ whose union covers $P$. 
Each such ball $B_j$, corresponding to an index $j$ (which is maintained in the index set $I_1$), is supposed to contain the optimal ball $B_j^\star$.  Further, the center $c_j$ of $B_j$ should have capacity at least that of $c_j^\star$. In each iteration of this {\bf while} loop, we first pick an  uncovered point $p$ and guess the index $j$ such that $p \in C_j^\star$ (line~\ref{l:guessj}). We now find the highest capacity point $x \in B(p, r_j)$ (while avoiding the centers already chosen).  We now call {\bf InsertBall}$(x,j, 3r_j, 1)$, i.e., we would like to insert the ball $B(x, 3r_j)$ to $\B_1$.

  In an iteration during lines~\ref{l:while3}--\ref{l:while4}, we add one ball to our solution for a remaining index $j \in I_4$. In particular, we pick the index $j \in I_4$ with the highest radius $r_j$ (line~\ref{l:max}) and guess a random subset $T_j$ of $I_1 \cup I_2$ (line~\ref{l:Tj}). The set $T_j$ is supposed to denote the indices $h \in I_1 \cup I_2$ such that the ball $B_h$ intersects the optimal ball $B_j^\star$. Assume that the algorithm guesses the set $T_j$ correctly (which again happens with probability at least $\frac{1}{2^k}$). Now two cases arise: (i) The radius $r_j$ is at least the radius of some ball $B_h, h \in T_j$, (ii) The radius $r_j$ is less than the latter quantity for all $h \in T_j$. Note that $\rad(B_h)$ is either $3r_h$ (when $h \in I_1$) or $5r_h$ (when $h \in I_2$). 

In the first case (line~\ref{l:first}), we identify a ball $B_j$ containing $B_j^\star$. Let $h$ be the index in $T_j$ such that $r_j$ is at least $\rad(B_h)$. 
We find the highest capacity point $x \in B(c_h, \rad(B_h) + r_j)$  (line~\ref{l:x}) except the already chosen centers -- since $B_h$ intersects $B_j^\star$, we note that $c_j^\star$ is a possible candidates for the point $x$, and hence $U_x \geq U_{c_j^\star}$
Finally, we add the ball $B_j = B(x, 5r_j)$ to $\B_2$ (line~\ref{l:addi2}). It is not difficult to show that $B_j$ contains $B_j^\star$. The second case, when $r_j < \rad(B_h)$ for all $h \in T_j$, is more challenging. We first identify a candidate set of points $P_j$ which contain the optimal ball $B_j^\star$ (line~\ref{l:Pj}) and then prune enough points from it to identify a ball $B_j$. 
For each $h \in T_j$ (line~\ref{l:Eh}), we define an {\em extended} ball $E_h$ of radius $9r_h$ around $c_h$ (clearly, $E_h$ contains the ball $B_h$, which is of radius either $3r_h$ or $5r_h$ around $c_h$). We shall show later $E_h$ contains $B_j^\star$.
Thus, $B_j^\star \subseteq \bigcap_{h \in T_j} E_h$. By our assumption on the correct guess of the set $T_j$, it follows that $B_i \cap B_j^\star$ is empty for $i \in (I_1 \cup I_2) \setminus T_j$
Further,  {\bf Invariant~1} implies that for any subset  $C_i$ in $\B_3$,  $i \in I_3$, $B_j^\star \cap C_i$ is empty. Thus, the set $P_j$ as defined in line~\ref{l:Pj} contains $B_j^\star$.  We now call the subroutine {\bf UpdateBalls} to add a suitable ball corresponding to $B_j^\star$ to $\B_3$ (or move a ball from $\B_3$ to $\B_2$).

\begin{algorithm}[H]
\caption{An iteration of the algorithm for $\capsum$ for general capacities.}
\label{algo:nonuni}
{\bf Input:} Set $P$ of $n$ points, parameter $k,$  radii $r_1, \ldots, r_k$, capacities $U_p$ for each point $p \in P$. \;
Initialize $I_1, I_2, I_3$ to emptyset and $I_4$ to $[k].$ \;
Initialize the sets $\B_1, \B_3, \B_3$  to emptyset. \;
\While{the balls in $\B_1$ do not cover $P$ \label{l:whileloop}}{ \label{l:while1}
Pick a point $p \in P \setminus \left( \bigcup_{j \in I_1} B_j \right)$. \label{l:p}\;
 \ (Output {\bf fail} if $I_4$ is empty). \label{l:fail}\;
Choose $j \in I_4$ uniformly at random. \label{l:guessj}\;
Let $x$ be the highest capacity point in $\left( P \setminus \{c_i : i \in I_1\} \right) \cap B(p,r_j).$ \label{l:choosex}\;
Call {\bf InsertBall}$(x,j, 3r_j, 1)$
\label{l:while2}
}
\While{$I_4$ is non-empty \label{l:whileempty}}{ \label{l:while3}
Let $j \in I_4$ with the highest $r_j$ value. \label{l:max} \;
Choose a random subset $T_j \subseteq I_1 \cup I_2$. \label{l:Tj} \;
\If{there is an index $h \in T_j$ with $r_j \geq \rad(B_h)$ \label{l:if}}{ \label{l:first}
Let $x$ be the highest capacity point (other than $\{c_i: i \notin I_4\}$) in $B(c_h, \rad(B_h)+r_j)$. \label{l:x}\;
Call {\bf InsertBall}($x, j, 5r_j, 2$).  \label{l:addi2}\;
}
\Else{ \label{l:second}
For an index $h \in T_j$, define $E_h := B(c_h, 9r_h)$. \label{l:Eh} \;
Define $P_j := \left( \bigcap_{h \in T_j} E_h \right) \setminus \left( \bigcup_{i \in (I_1 \cup I_2) \setminus T_j} B_i \cup \bigcup_{i \in I_3} C_i \right).$ \label{l:Pj}\;
Call {\bf UpdateBalls}($j, P_j$). \label{l:while4}
}
}
Output the balls $\{ B(c_j, 9 r_j) : j \in [k] \} \label{l:output}$
\end{algorithm}

We  now give details of the {\bf UpdateBalls} procedure. We also emphasize that this is step where the technical novelty of our contribution lies.  We assume that the sets $I_1, \ldots, I_4$, radii $r_1, \ldots, r_k$ and $\B_1, \ldots, \B_3$ can be accessed or modified by this procedure. The parameters given to this procedure are an index $j \in I_4$ and a set $P_j$ of points which should contain the optimal ball $B_j^\star$. The procedure can terminate in the following manner: (i) find an index $i \in I_3$ and a ball containing $B_i^\star$ -- in this case, we move $i$ from $I_3$ to $I_2$ (and change $\B_2, \B_3$ accordingly), or (ii) corresponding to the index $j$, add a ball $B_j$ and a subset $C_j \subseteq B_j$ to $\B_3$. Note that~{\bf Invariant~1} requires that $C_j$ should be disjoint from $B_i^\star$ for all $i \in I_3 \cup I_4$. Therefore the algorithm maintains a set $Z$ (initialized to $I_3 \cup I_4$) of indices that could potentially intersect the intended ball $B_j$ (line~\ref{l:Z}). The set $Z$ gets pruned as we run through the iterations of the {\bf while} loop (lines~\ref{l:w5}--\ref{l:updateZ}). We now describe each iteration of this {\bf while} loop. In line~\ref{l:x1}, we identify a point $x$ of a high capacity such that $B(x, r_j)$ has a large intersection with $P_j$; and we let $C_j$ denote this intersection. Now, we guess the subset $L$ of $Z$ consisting of the indices $i \in Z$ such that $B_i^\star \cap C_j$ is  not empty (line~\ref{l:L}). Assume that this guess is correct.  If $L$ is empty (line~\ref{l:Le}), we can add $B_j$ to $\B_3$. The procedure terminates in this case. 

Hence assume that the set $L$ is not empty. There are two sub-cases now. In the first sub-case (line~\ref{l:Ls1}), there is an index $t \in L$ such that $r_t > r_j$. We shall show that $t \in I_3$ (recall that $L$ is a subset of $I_3 \cup I_4$). Since $B_j$ and $B_t^\star$ intersect, the ball $B(x, r_t + r_j)$ contains $c_t^\star$. Thus, we pick a high capacity point $y$ in this ball (line~\ref{l:Y}) and add the ball $B(y,5r_t)$ to $\B_2$ (line~\ref{l:insertB2}) -- the fact that $r_t > r_j$ ensures that this ball contains $B_t^\star$. Thus, the index $t$ moves from $I_3$ to $I_2$. The procedure terminates in this sub-case. 

Finally, consider the sub-case when the set $L$ is non-empty and every index $t \in L$ satisfies $r_t \leq r_j$. In line~\ref{l:Ls2}, we guess whether the ball $B_j' := B(x, 3r_j)$ intersects $B_j^\star$.  Assume that this guess is correct. If $B_j' \cap B_j^\star$ is non-empty (line~\ref{l:add2}), then it is easy to see that the ball $B(x, 5r_j)$ contains $B_j^\star$.  Thus, we add this ball to $\B_2$ and terminate. Otherwise, we can remove $B_j$ from $P_j$ (recall that $P_j$ contains a subset of points which are guaranteed to contain $B_j^\star$). Now that $P_j$ has shrunk (and hence, the possible set of points in the desired ball $B_j$ also reduces), we can update the set $Z$. Recall that $Z$ stores indices $i \in I_3 \cup I_4$ such that $B_i^\star$ potentially intersects $B_j^\star$. Since $r_t \leq r_j$ and $B_t^\star$ intersects $B(x, r_j)$ for all $t \in L$, it follows that $B_t^\star \subseteq B_j'$. Thus, $B_t^\star$ does not intersect the updated set $P_j$. Thus, we can remove the $L$ from $Z$ (line~\ref{l:updateZ}). This is the only sub-case where we perform another iteration of the {\bf while} loop. Since we reduce the size of $Z$ by at least 1, there can be at most $k$ iterations of the {\bf while} loop. 
This completes the description of our algorithm. 

\begin{algorithm}[ht]
\caption{Procedure {\bf UpdateBalls}($j,P_j$)  adds a ball to either $\B_2$ or $\B_3$.}
\label{algo:update}
{\bf Input:}  An index $j \in I_4$, a subset $P_j$ of points. \;
Initialize a variable {\tt update} to false \;
Initialize an index set $Z := I_3 \cup I_4$. \label{l:Z} \;
\While{{\tt update} is false \label{l:w5}}{
 Let $x$ be the point in $P \setminus \{c_i: i \notin I_4\}$ which maximizes $\min(U_x, |B(x, r_j) \cap P_j|)$. Define $C_j := B(x, r_j) \cap P_j$. \label{l:x1} \;
 Pick a subset $L \subseteq Z$ uniformly at random. \label{l:L} \;
 \If{$L$ is empty \label{l:Le}}{
    Call {\bf InsertBall}($x, j, r_j, 3$). \label{l:done1} \;
    If $B(x, r_j)$ was added to $\B_3$, set $C_j$ as defined above in line~\ref{l:x1}. \;
    Set {\tt update} to true. 
 }
 \Else{ 
   \If{there is an index $t \in L$ with $r_t > r_j$ \label{l:Ls1}}{
        
        Let $y$ be the point in $B(x, r_j + r_t) \setminus \{c_i: i \notin I_4\}$ with the maximum capacity $U_y$. \label{l:Y} \;
        Call {\bf InsertBall}($y, t, 5r_t, 2$); \label{l:insertB2} \;
        If the index $t$ was added to $I_2$, remove it from $I_3$ and remove the corresponding ball from $\B_3$. \label{l:shift} \;
        Set {\tt update} to true. 
   }
   \Else{ \label{l:Ls2}
        Perform exactly one of the following steps with equal probability: \;
        (i) 
        Call {\bf InsertBall}($x, j, 5r_j, 2$). Set {\tt update} to true. \label{l:add2} \;
        (ii) Update $Z = Z \setminus L$ and $P_j = P_j \setminus B(x, 3r_j).$ \label{l:updateZ}
   } 
 }
}
\end{algorithm}

\subsection{Analysis}
We now analyse the algorithm. We begin with the following key observation about the procedure {\bf InsertBall}:

\begin{claim}
    \label{cl:insertball}
    Suppose the algorithm calls {\bf InsertBall} with parameters $(p, j, r, u)$, then the point $p \neq c_i$ for any $i \in I_1 \cup I_2 \cup I_3$. 
\end{claim}
\begin{proof}
    The claim follows from the fact that whenever we call 
    {\bf InsertBall}$(p,j,r,u)$, we make sure $p$ has not been selected as the center of any ball in $\B_1 \cup \B_2 \cup \B_3$. This can be easily seen by considering each call to this procedure: 
    \begin{itemize}
        \item Line~\ref{l:while2} of~\Cref{algo:nonuni}: At this point $I_2, I_3$ are empty and we make sure that  the point $x \neq c_i$ for any $i \in I_1$ (line~\ref{l:choosex}).
        \item Line~\ref{l:addi2} of~\Cref{algo:nonuni}: we ensure that $x \neq c_i, i \in [k] \setminus I_4$ in line~\ref{l:x}. 
        \item Line~\ref{l:done1} or line~\ref{l:updateZ} of~\Cref{algo:update}: in line~\ref{l:x1}, we make sure that $x \neq c_i, i \in [k] \setminus I_4$. 
        \item Line~\ref{l:insertB2} of~\Cref{algo:update}: we ensure in line~\ref{l:Y} that $y \neq c_i, i \in [k] \setminus I_4$. 
    \end{itemize}
\end{proof}

We now bound the number of iterations in the procedure~{\bf UpdateBalls}. 

\begin{claim}
    \label{cl:iter}
    The {\bf while} loop in a particular invocation of~{\bf UpdateBalls} has at most $k$ iterations. 
    The procedure {\bf InsertBall} is called at most $O(k^2)$ times, and the number of times it chooses the option (ii) is at max $k$.
\end{claim}

\begin{proof}

    First consider the {\bf while} loop in the {\bf UpdateBalls} procedure. If a particular iteration of this loop does not end the procedure, then it must execute line~\ref{l:updateZ}. Since the set $L \subseteq Z$ is non-empty, the set $Z$ reduces in size during this iteration. Since the initial size of $Z$ was at most $k$, there can be at most $k$ iterations of this {\bf while loop}. 
    
    The number of times {\bf InsertBall} chooses the  option (ii) is at max $k$, because every time the second option is chosen, the size of $I^\star$ increases by $1$, and we never remove anything from $I^\star$. During any consecutive stretch of option (i) choices, each call to {\bf InsertBall}, either removes an element from $I_4$ or moves an element from $I_3$ to $I_2$ (line~\ref{l:shift}). Hence, the number of consecutive calls to {\bf InsertBall} that make the first choice can be at max $2k$. 
\end{proof}

    

%

We first state the desirable events during an iteration of~\Cref{algo:nonuni}:
\begin{itemize}
    \item $\cE_1$: For each point $p$ considered in line~\ref{l:p}, the index $j$ chosen in line~\ref{l:guessj} satisfies the condition that $p \in C^\star_j$. 
    \item $\cE_2$: For each index $j$ considered in line~\ref{l:max}, the subset $T_j$ chosen in line~\ref{l:Tj} satisfies the property that $T_j = \{h \in I_1 \cup I_2: B_h \cap B_j^\star \neq \emptyset \}.$
    \item $\cE_3:$ For each iteration of the {\bf while} loop in~\Cref{algo:update}, the subset $L$ picked in line~\ref{l:L} in~\Cref{algo:update} is equal to the set $\{i \in Z: B_i^\star \cap C_j \neq \emptyset\}$. 
    \item $\cE_4$: In each iteration of the {\bf while} loop in~\Cref{algo:update}, if we reach line~\ref{l:Ls2}, then the choice~(i) is taken iff the ball $B(x,3r_j) \cap B_j^\star$ is non-empty.  
    \item $\cE_5$: Whenever {\bf InsertBall}($p,j,r,u$) is called, it chooses option~(ii) in line~\ref{l:step2} iff $p$ is same as $c_u^\star$ for some $u \in I_2 \cup I_3 \cup I_4$. Further, the index $h$ selected in this step is equal to $u$. 
\end{itemize}

\begin{claim}
    \label{cl:event1}
    If the event $\cE_1 \cap \cE_5$ happens, then for each $j \in I_1$, the ball $B_j \in \B_1$ contains the optimal ball $B_j^\star$. Hence, if $\cE_1$ happens, the algorithm does not output {\bf fail} in line~\ref{l:fail}. Further, the probability that all the events $\cE_1, \cE_2, \cE_3, \cE_4, \cE_5$ happen is at least $1/2^{O(k^4)}$.
\end{claim}
\begin{proof}
    The proof is similar to that of~\Cref{lem:event}.
    Suppose $\cE_1 \cap \cE_5$ happens. For each index $i \in I^\star$, $B(c_i, r_i)$ contains $B_i^\star = B(c_i^\star, r_i^\star)$ because $c_i = c_i^\star$ for such indices and $r_i \geq r_i^\star$. Now if {\bf InsertBall} executes option (ii) in line~\ref{l:step2}, we simply replace $I_1$ by $I^\star$ and start from scratch. The only other case when when we add a ball to $\B_1$ is during the {\bf while} loop in line~\ref{l:whileloop} in~\Cref{algo:nonuni}. Consider such an iteration and assume that {\bf InsertBall} executes option (i). Let $x$ be the point chosen in line~\ref{l:choosex}. 
    Then,  $d(x, c_j^\star) \leq d(x, p) + d(p, c_j^\star) = 2r_j$. Therefore, $B_j := B(x, 3r_j)$ contains $B_j^\star$. 
    Therefore, there will be at most $k$ iterations of the {\bf while} loop in line~\ref{l:whileloop}. Thus, the algorithm won't output {\bf fail} in line~\ref{l:fail}.

 The probability that {\bf InsertBall} always chooses the option (i) or (ii) correctly is at least $ \frac{1}{2^{O(k^2)}}$, as the number of calls is $O(k^2)$ (~\Cref{cl:iter}) and the probability of making the correct choice in a single call is $\frac{1}{2}$. Given this happens, the probability that it correctly guesses the index $h$ is at least $1/k$ for any call where it chooses option (ii). Hence, $\Pr[\cE_5] \geq \frac{1}{k^k} \frac{1}{2^{O(k^2)}} = \frac{1}{2^{O(k^2)}}$.
    Now, the probability that the chosen index $j$ in line~\ref{l:guessj} is correct (i.e., $p \in C_j^\star$)  in each iteration of this {\bf while} loop is at least $1/k$. Hence, $\Pr[\cE_1 | \cE_5] \geq 1/k^{O(k^2)}$, as each such iteration leads to a call to {\bf InsertBall}. Conditioned on $\cE_1$, the probability that the choice of $T_j$ in line~\ref{l:Tj} is correct is at least $1/2^k$ (since there are at most $2^k$ possibilities for $T_j$). Since every iteration of the {\bf while} loop in line~\ref{l:whileempty} leads to a call to {\bf InsertBall}, we see that $\Pr[\cE_2|\cE_1,\cE_5] \geq \frac{1}{2^{O(k^3)}}.$

    Given events $\cE_1, \cE_2$, whenever we reach line~\ref{l:L}, we guess $L$ correctly with probability at least $\frac{1}{2^k}$. There are at most $O(k^2)$ calls to {\bf UpdateBalls}, and there are at most $k$ iterations of the {\bf while} loop in~\Cref{algo:update}. Hence, $\Pr[\cE_3 | \cE_1, \cE_2, \cE_5] \geq \frac{1}{2^{O(k^4)}}.$ Also, whenever reach line~\ref{l:Ls2}, we guess the choice (i) or (ii) correctly with probability $1/2$. Hence, $\Pr[\cE_4 | \cE_1, \cE_2, \cE_3, \cE_5] \geq \frac{1}{2^{O(k^3)}}$
\end{proof}

Now we write down the invariant conditions satisfied during our algorithm. The condition \inv{1} was mentioned earlier, but we restate it here:
\begin{itemize}
    \item \inv{1}: Let $j, j' \in I_3$ be two distinct indices. Then $C_j \cap C_{j'} = \emptyset$. Further, for any $j \in I_3$ and $i \in I_3 \cup I_4, $ $B_i^\star \cap C_j = \emptyset.$
    \item \inv{2}:   For any index $j \in I_1$, $B_j^\star \subseteq B_j$, where $B_j$ is the ball corresponding to index $j$ in $\B_1$.  Further, the radius of $B_j$ is at most $3r_j$. 
    \item \inv{3}: For any index $j \in I_2$, $B_j^\star \subseteq B_j$, where $B_j$ is the ball corresponding to index $j$ in $\B_2$.  Further, the radius of $B_j$ is at most $5r_j$. 
    \item \inv{4}: For any index $j \in I_4$, and any index $i \in I_2$, $r_j \leq r_i$ . 
\end{itemize}

We now show that these invariant conditions are maintained by the algorithm. The following result follows from the proof of~\Cref{cl:event1}. 
\begin{claim}
    \label{cl:inv2}
    Suppose event $\cE_1$ happens. Then \inv{2} is maintained by the algorithm. 
\end{claim}

\begin{lemma}
\label{lem:invr}
    Suppose events~$\cE_1, \ldots, \cE_5$ occur. Then the 
    algorithm maintains all the four invariant conditions mentioned above. 
\end{lemma}
\begin{proof}
    Assume that the events $\cE_1, \ldots, \cE_5$ occur. 
    \Cref{cl:inv2} already shows that \inv{2} is maintained by our algorithm. We now show that the other three invariant conditions hold by induction on the number of iterations of the {\bf while} loop in line~\ref{l:whileempty} in~\Cref{algo:nonuni}. Just before executing this {\bf while} loop for the first time, the index sets $I_2, I_3$ are empty, and so the three invariant conditions hold vacuously. 

     Now consider a particular iteration of this {\bf while} loop and assume that the invariant conditions hold at the beginning of this iteration. During this {\bf while} loop, we shall call {\bf InsertBall} procedure exactly once. In this procedure, if the second option is chosen (i.e., line~\ref{l:step2}), then we shall insert a ball $B_h$ in $\B_1$ and an index $h$ in $I_1$ such that $B_h$ contains $B_h^\star$ (since $\cE_5$ occurs). Now it is easy to check that all the invariants continue to hold. Therefore, we shall assume that this procedure executes the first option in line~\ref{l:step1}.

     Now, two cases arise depending on whether the condition in line~\ref{l:if} is true. First assume it is true, i.e., there is an index $h \in T_j$ such that $r_j \geq  \rad(B_h)$. Let $x$ be the point selected in line~\ref{l:x}. Then 
     $$d(x, c_{j^\star}) \leq d(x, c_h) + d(c_h, c_j^\star) \leq  \rad(B_h) + r_j + \rad(B_h) + r_j \leq 2 \cdot \rad(B_h) + 2r_j,$$
     where we have used the fact that $d(c_h, c_j^\star) \leq \rad(B_h) + r_j^\star$ because $B_h$ and $B^\star_j$ intersect. Since $r_j \geq \rad(B_h)$, the above is at most $4r_j$. It follows that $B(x, 5r_j)$ contains $B_j^\star$. This shows that \inv{3} is satisfied. Since $j$ was chosen to have the highest $r_j$ value among all indices in $I_4$, \inv{4} is also satisfied. Finally, we do not change $I_3$ and only remove an index from $I_4$. Therefore, \inv{1} continues to hold. 

     We now consider the more involved case when the outcome of the {\bf if} condition in line~\ref{l:if} is false.
    In this case, $r_j < \rad(B_h)$ for all $h \in T_j$. We first argue that $B_j^\star \subseteq E_h$ for any $h \in T_j$ (where $E_h = B(c_h, 9r_h)$ as defined in line~\ref{l:Eh}). To see this, note that, $B_j^\star \subseteq B(c_h, \rad(B_h) + 2 r_j)$,  as $B_j^\star$ and $B_h$ intersect. If $h \in I_1$, from \inv{2}, $\rad(B_h) \leq 3 r_h$, so $\rad(B_h) + 2 r_j \leq 3 \rad(B_h) \leq 9 r_h$. Otherwise, $\rad(B_h) \leq 5 r_h$ from \inv{3}, and $r_j \leq r_h$ from \inv{4}. Hence, $\rad(B_h) + 2 r_j) \leq 5 r_h + 2 r_h \leq 9 r_h$ 
    Thus, $B_j^\star \subseteq \cap_{h \in T_j} E_h$. Since the event $\cE_2$ has occurred, we know that $B_j^\star \cap B_h = \emptyset $ for all $h \in I_1 \cup I_2 \setminus T_j$, and the induction hypothesis about \inv{1} implies that $B_j^\star \cap C_i = \emptyset$ for all $i \in I_3$. Thus, the set $P_j$ defined in line~\ref{l:Pj} of~\Cref{algo:nonuni} contains $B_j^\star$. 
    
    \begin{claim}
        \label{cl:Z}
        $B_j^\star \subseteq P_j$ during the execution of the procedure {\bf UpdateBalls}($j,P_j)$. Further, for every  index $i \in (I_3 \cup I_4) \setminus Z$,   $B_i^\star \cap P_j = \emptyset$. 
    \end{claim}
    \begin{proof}
        We show this by induction on the number of iterations of the {\bf while} loop in the procedure {\bf UpdateBalls}($j,P_j)$. Since $Z$ is initialized to $I_3 \cup I_4$, and $B_j^\star \subseteq P_j$ when this procedure is called, the claim is true at the beginning. 
        
         Now suppose the claim is true at the beginning of an iteration of this {\bf while} loop. Assume that  we go through one iteration of the {\bf while} loop and come back to line~\ref{l:w5}.
          Then we must have executed the case in line~\ref{l:updateZ}. Since event $\cE_3 \cap \cE_4$ happens, we know that $L = \{i \in Z: B_i^\star \cap C_j \neq \emptyset\}$, $r_t \leq r_j$ for all $t \in L$ and $B(x, 3r_j) \cap B_j^\star$ is empty. Therefore, $B_j^\star \subseteq P_j' := P_j \setminus B(x, 3r_j)$. We claim that for any $t \in L$, $B_t^\star \cap P_j'$ is empty. Indeed, we know that $B_t^\star$ intersects $C_j$ (by event $\cE_4$) and $r_t \leq r_j$. Since $C_j \subseteq B(x, r_j)$, it follows that $B_t^\star \subseteq B(x, 3r_j)$. Therefore, $B_t^\star$ does not intersect $P_j'$. Thus, we see that the desired claim holds at the end of line~\ref{l:updateZ} as well.            
    \end{proof}
    Now we consider the various cases on how this procedure terminates:
    \begin{itemize}
        \item The ball $B(x,r_j)$ was added in line~\ref{l:done1}: Since event $\cE_3$ occurred, we know that $B_i^\star \cap C_j$ is empty for all $i \in Z$.   Combining this with~\Cref{cl:Z} and the fact that $C_j \subseteq P_j$, we see that $B_i^\star \cap C_j$ is empty for all $i \in I_3 \cup I_4$. Since $P_j$, as defined in line~\ref{l:Pj}, is disjoint from $\cup_{i \in I_3} C_j$, we see that adding the index $j$ to $I_3$, and $B(x,r_j)$ along with $C_j$ to $\B_3$ maintains \inv{1}.  Since we do not change $I_1, I_2$ and $I_4$ only reduces, the other three invariants continue to be satisfied. 
        \item the ball $B_t := B(y, r_j+r_t)$ was added to $\B_2$ in line~\ref{l:insertB2} (and the index $t$ and the corresponding ball was removed from $\B_3$): Since we are removing a ball from $\B_3,$ \inv{1} continues to be satisfied. No new ball is added to $\B_1$ and hence, \inv{2} is also satisfied. We check \inv{3} now: we know that
        $r_t > r_j$ and $B_t^\star$ intersects $B_j$ (since $\cE_3$ has occurred).  Therefore, 
        $$d(y, c_t^\star) \leq d(y,x) + d(x, c_t^\star) \leq (r_j + r_t) + (r_j + r_t). $$
        Since $r_j \leq r_t$, the above is at most $4r_t$. Therefore, $B(y, 5r_t)$ contains $B_t^\star$. Thus, \inv{3} is satisfied. It remains to check \inv{4}: note that the index $j \in I_4$ was chosen because it has the highest $r_j$ value among all the indices in $I_4$ (line~\ref{l:max}). Since $r_t \geq r_j$, we see that \inv{4} is also satisfied. 
        \item The ball $B(x, 5r_j)$ is added to $\B_2$ in line~\ref{l:add2}: Since $I_1, I_3$ do not change and the set $I_4$ shrinks, invariants~\inv{1}, \inv{2} continue to be satisfied. Since $\cE_4$ occurs, we know that $B(x, 3r_j) \cap B_j^\star$ is non-empty. Hence, $B(x, 5r_j)$ contains $B_j^\star$ and so, \inv{3} is satisfied. Finally, \inv{4} is satisfied because of the manner index $j \in I_4$ was chosen (line~\ref{l:max}). 
    \end{itemize}
This completes the proof of the desired lemma. 
\end{proof}

We now show that the centers of the balls added by our algorithm have sufficient capacities: 
\begin{lemma}
    \label{lem:capacity}
  Assume that the events $\cE_1, \ldots, \cE_5$ occur.   For each ball $B_j \in \B_1 \cup \B_2 \cup \B_3$ centered at $c_j$, the capacity of $c_j$ is at least $|C_j^\star|$. 
\end{lemma}
\begin{proof}
    We first observe that the following property is maintained during the execution of the algorithm: 
    \begin{claim}
        \label{cl:center}
        For any index $i \in I_4$, the optimal center $c_i^\star$ is not the center of any ball in $\B_1 \cup \B_2 \cup \B_3.$   
    \end{claim}
    \begin{proof}
       We add a ball to $\B_2 \cup \B_3$ only during case~(i) of~{\bf InsertBall}. But if $p$ happens to be $c_i^\star$ for some $i \in I_2 \cup I_3 \cup I_4$, then $\cE_5$  implies that we invoke case~(ii), a contradiction. Therefore, the claim holds.             
    \end{proof}
     
    When we invoke case~(ii) of {\bf InsertBall} and a ball $B_j$ to $\B_1$, \inv{5} implies that the center of this ball is $c_j^\star$ and hence, the Lemma holds trivially. Therefore, for rest of the proof, consider only the cases when we execute option~(i) (i.e., line~\ref{l:step1}) when the {\bf InsertBall} procedure is called. We now go through all such possibilities: 
    \begin{itemize}
        \item We add the ball $B(x, 3r_j)$ to $\B_1$ in line~\ref{l:while2}: we know by the property shown above that $c_j^\star$ has not been used as the center of any selected ball yet. Since $p \in B_j^\star$ (by event $\cE_1$), $c_j^\star$ is a candidate for the point $x$ selected in line~\ref{l:choosex}. Therefore, $U_x \geq U_{c_j^\star} \geq |C_j^\star|.$
        \item We insert the ball $B(x, 5r_j)$ in line~\ref{l:addi2}: again, we know  that $c_j^\star$ has not been used as a center yet. Since $B_h \cap B_j^\star$ is non-empty (since event $\cE_2$ occurs), $d(c_h, c_j^\star) \leq \rad(B_h) + r_j$. Therefore, $c_j^\star$ is one of the candidates for the point $x$. Thus, $U_x \geq U_{c_j^\star} \geq |C_j^\star|.$
        \item We add a ball centered at $x$ in line~\ref{l:done1} or line~\ref{l:updateZ} in {\bf UpdateBalls} procedure: We know by~\Cref{cl:Z} that $B_j^\star \subseteq P_j$. Therefore, for $x = c_j^\star$, the quantity $\min(U_x, |B(x, r_j) \cap P_j|)$ is at least $|C_j^\star|$. Since $c_j^\star$ has not been chosen as a center yet, it follows from line~\ref{l:x1} that $U_x \geq |C_j^\star|$. 
        \item We add the index $t$ to $I_2$ and a ball centered at $y$ to $\B_2$ in line~\ref{l:Y}: Since $B_j \cap B_t^\star \neq \emptyset$ (by event $\cE_3$), $d(x,c_t^\star) \leq r_j  + r_t$. 
        We claim that $c_t^\star$ has not been used a center of any ball in $\B_1 \cup \B_2 \cup \B_3$. Indeed, we know that $t \in I_3 \cup I_4$. If $t \in I_4$,~\Cref{cl:center} implies that $c_t^\star$ has not been used as a center. 
        If $t \in I_3$, then \inv{1} implies that $c_t^\star$ is not contained in the ball $B_t \in \B_3,$ and hence, has not been used a center yet. Since $d(x, c_t^\star) \leq r_j + r_t$, $c_t^\star$ is a candidate for the point $y$ chosen in line~\ref{l:Y}. Therefore, $C_y \geq C_{c_t^\star} \geq |C_t^\star|$.    
    \end{itemize} 
\end{proof}

\begin{lemma}
    \label{lem:feas}
    Suppose events $\cE_1, \ldots, \cE_5$ occur. Let $\{B(c_j, 9 r_j) : j \in [k]\}$ be the set of balls output by~\Cref{algo:nonuni}.  Then there is a feasible solution to this instance where each center $c_j$ is assigned at most $U_j$ points, and if a point $p$ is assigned to a center $c_j$, then $p \in B_j$. 
\end{lemma}
\begin{proof}
    For the balls in $\B_1$ and $\B_2$, \inv{2} and \inv{3} shows that they cover the corresponding optimal balls. However this is not true for the balls in $\B_3$ (in fact \inv{1} shows quite the opposite). Thus handling these balls require more nuanced argument. 
    For each $j \in I_3$, we claim that $|C_j| \geq |C_j^\star|$. Indeed, as the proof of~\Cref{lem:capacity} shows, when we define $C_j$ as in line~\ref{l:x1}, $c_j^\star$ is also a candidate for the point $x$ chosen in this step. Therefore, 
    $\min(U_{c_j}, |C_j|) \geq \min(U_{c_j^\star}, |C_j^\star|)$, which implies that $|C_j| \geq |C_j^\star|$. 
    Therefore, for each $j \in I_3$, let $D_j$ be an arbitrary subset of $C_j$ with $|D_j| = |C_j^\star|$. 
    It follows from \inv{1} that $D_i \cap D_j = \phi$ for any $i, j \in I_3, i \neq j$. For any $j \in I_3, i \in T_j$, define $S_{ji} := D_j \cap C_i^\star$. Note that $\cup_{i \in T_j} S_{j i} = D_j$ because $T_j$ is the set of indices $i$ for which $C_j \cap C_i^\star$ is non-empty. Also, the sets $S_{ji}$ are disjoint as $C_i^\star$ are disjoint, and therefore $\sum_{i \in T_j} |S_{ji}| = |D_j| = |C_j^{\star}|$. Hence, we can partition the set $C_j^\star$ as $C_j^\star = \cup_{i \in T_j} V_{ji}$, where $V_{ji}$'s are mutually disjoint over different $i$ and $|V_{ji}| = |S_{ji}|$.

    Now, let us define the clusters $X_1, X_2, \ldots X_k$ as follows:
    \begin{itemize}
        \item For each $i \in I_3$, $X_i = D_i$
        \item For each $i \in I_1 \cup I_2$, $X_i = (C_i^\star \setminus \cup_{j: i \in T_j} S_{ji}) \cup (\cup_{j : i \in T_j} V_{ji})$. Note the following properties:
            \begin{itemize}
                \item $S_{ji} = C_i^\star \cap D_j \subseteq C_i^\star$
                \item $V_{ji} \subseteq C_j^\star$ meaning $V_{ji} \cap C_i^\star = \phi$
                \item For for two indices $x, y \in I_3$ such that $i \in T_{x} \cap T_{y}$, $S_{x i} \cap S_{y i} = \phi$, as $D_{x} \cap D_{y} = \phi$ and $V_{x i} \cap V_{y i} = \phi$, as $C_{x}^\star \cap C_{y}^\star = \phi$
            \end{itemize}
            Hence, $|X_i| = |C_i^\star| - \sum_{j : i \in T_j} |S_{ji}| + \sum_{j: i \in T_j} V_{ji} = C_i^\star$ as $|S_{ji}| = |V_{ji}|$ for $j$ such that $i \in T_j$.
    \end{itemize}

    We complete the proof by claiming that $X$ is a valid assignment, that is:
    \begin{enumerate}
        \item For each $i \in [k]$, $|X_i| \leq U_{c_i}$.
        \item For each $i \in [k]$, $X_i \subseteq B(c_i, 9 r_i)$
        \item $X_i \cap X_j = \phi$ for any $i \neq j \in [k]$ and $\cup_{i \in [k]} X_i = P$.
    \end{enumerate}

    The first part can be proved by noting that for all $i \in [k]$, $|X_i| = |C_i^\star|$ and then using ~\Cref{lem:capacity}. 
    
    For the second part, if $i \in I_3$, $X_i = D_i \subseteq C_i \subseteq B_i = B(c_i, r_i) \subseteq B(c_i, 9 r_i)$. Otherwise, note that $C_i^\star \subseteq B_i^\star \subseteq B(c_i, 5 r_i) \subseteq B(c_i, 9r_i)$ holds from \inv{3}. Also, for any $j$ with $i \in T_j$, note that $V_{ji} \subseteq C_j^\star \subseteq P_j \subseteq B(c_i, 9r_i)$. 
    
    For the third part, consider $i \neq j \in [k]$. There are three cases:
    \begin{itemize}
        \item $i, j \in I_3$: In this case, $X_i \cap X_j = D_i  \cap D_j = \phi$
        \item $i, j \in I_1 \cup I_2$: In this case, $C_i^\star \cap C_j^\star = \phi$, and for any two indices $x, y$, with $i \in T_{x}$ and $j \in T_{y}$, $V_{x i} \cap V_{y j} = \phi$ because if $x \neq y$, then $V_{x i} \subseteq C_{x}^\star$ and $V_{y j} \subseteq C_{y}^\star$, and otherwise if $x = y$, $V_{x i}$ and $V_{y j}$ are disjoint sets in a partition of $C_{x}^\star$
        \item $i \in I_1 \cup I_2, j \in I_3 \cup I_2$ or vice versa. Assume wlog that  $i \in I_1 \cup I_2, j \in I_3$. Here, $D_j \cap (C_i^\star \setminus \cup_{l: i \in T_l} S_{li}) = \phi$ because we remove $S_{ji} = D_j \cap C_i^\star$ from $C_i^\star$ when forming $X_i$. Also, $D_i \cap V_{li} = \phi$ for any $l$ with $i \in T_l$, because $V_{l i} \subseteq C_l^\star$, $D_j \subseteq C_j$ and $C_j \cap C_l^\star = \phi$, from \inv{1}.

        Now, since the clusters are disjoint and $|X_i| = |C_i^\star|$ for all $i \in [k]$, $|\cup_{i \in [k]} X_i| = \sum_{i \in [k]} |X_i| = \sum_{i \in [k]} |C_i^\star| = |P|$. This means that $\cup_{i \in [k]} X_i = P$, as $X_i \subseteq P$ for all $i \in [k]$ 
    \end{itemize}
\end{proof}
\begin{lemma}
    \label{lem:cost}
 The total cost of the solution produced by~\Cref{algo:nonuni} is at most $9 \sum_{j \in [k]} r_j$. Further, the total $L_p$ norm of the radii of the clusters produced by this algorithm is at most $9 \left( \sum_{j \in [k]} r_j^p \right)^{1/p}$. 
\end{lemma}
\begin{proof}
    Clearly, the sum of $p^{th}$ powers of the radii of the balls in the algorithm's output is: $\sum_{j \in [k]} (9 r_j)^p = 9^p \sum_{j \in [k]} r_j^p$.
\end{proof}

It follows from  ~\Cref{cl:event1}, ~\Cref{lem:feas} and ~\Cref{lem:cost} that the algorithm outputs a $(9+\varepsilon)$ approximate solution with probability at least $1/2^{O(k^4)}$. This probability bound can be improved to $1/2^{O(k^3)}$ through a more fine-grained analysis. For this, note that the invariants $1,3$ and $4$ only need to be satisfied after the last call to {\bf InsertBall} that chooses option (ii). This allows us to relax the definitions of $\cE_2,\cE_3,\cE_5$, such that $\Pr[\cE_1 \cap \cE_2 \cap \cE_3 \cap \cE_4 \cap \cE_5] \geq 1 / 2^{O(k^3)}$

\subsection{Improvement in approximation ratio from $9$ to $4 + \sqrt{13}$}
In this section, we briefly show that with a more careful choice of parameters, we can improve the approximation ratio from 9 to about 7.606. We use a parameter $\alpha$ that shall be fixed later. We make the following changes in~\Cref{algo:nonuni}: 
\begin{itemize}
    \item The condition in line~\ref{l:if} becomes ``$r_j \geq \alpha \cdot \rad(B_h)$''. 
    \item Line~\ref{l:x} remains unchanged, i.e., $x$ is the maximum capacity point in  $B(c_h, \rad(B_h) + r_j)$ (other than $\{c_i: i \notin I_4\}$). But line~\ref{l:addi2} changes to ``Call {\bf UpdateBalls}$(x,j,2 \cdot \rad(B_h)  + 3r_j,2)$.''
    \item The definition of the ball $E_h$ in line~\ref{l:Eh} changes. Recall that $h$ belongs to the set $T_j$ here. Now we differentiate between two cases. When $h \in I_1$, then we define $E_h$ as $B(c_h, 3(1 + 2 \alpha) r_h)$. The other case is when $h \in I_2$. In this case, we define $E_h$ as $B(c_h, (5 + \frac{2}{\alpha}) r_h).$
    \item In line ~\ref{l:output}, we now output $\{B(c_j, 3(1 + 2 \alpha) r_j) : j \in I_1\} \cup \{B(c_j, (5 + \frac{2}{\alpha}) r_j) : j \in I_2\} \cup \B_3$
\end{itemize}

A routine modification of the analysis in the previous section shows that the approximation ratio now becomes  $\max(3(1 + 2 \alpha), 5 + \frac{2}{\alpha})$ , whose minimum possible value is $4 + \sqrt{13} \approx 7.606$ at $\alpha = \frac{1 + \sqrt{13}}{6}$.

%% file: 4-Hardness.tex
\section{Hardness of Approximation for \capsum}
\label{sec:eth}

In this section, we prove~\Cref{thm:eth}. This result states that assuming the \eth (\ETH), there exists  a constant $c > 1$ such that any $c$-approximation for \\
$\capsum$ must have running time exponential in $k$.
The proof of this result relies on reduction from the \vc problem on constant degree graphs. We shall use the following result that follows from Dinur's PCP Theorem~\cite{dinur}. Some of the details can be found in \cite{abjk18}.

\begin{theorem}
\label{thm:eth1}
Assume that \ETH holds. 
Then there exist constants $\veps, d > 0$ such that the following holds: any algorithm for \vc on graphs  with maximum vertex degree at most $d$ that  distinguishes between the case when the input graph $G$ has a vertex cover of size at most $\frac{2n}{3}$ and the case when $G$ has a vertex cover of size at least $\frac{2n}{3} \cdot (1 + \veps)$, must have running time $2^{\Omega \left( \frac{n}{\polylog n}\right)}$.
\end{theorem}



Let $d$ be the constant stated in~\Cref{thm:eth1}. Consider an instance $\cI$ of $\vc$. Let the instance $\cI$ be specified by a graph $G=(V,E)$, where each vertex has degree at most $d$. We now construct an instance $\cI'$ of \capsum as follows (see~\Cref{fig:red}): 

\begin{figure}
\begin{center}
\includegraphics[scale=0.5]{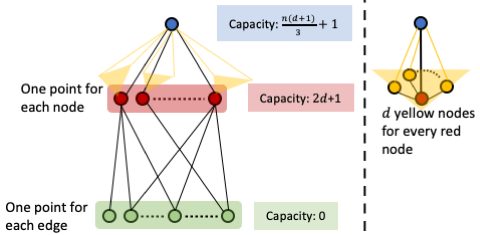}
\end{center}
\caption{Reduction from \vc to \capsum: points of the form $p_e$ are colored green, points $p_v$ are colored red. The subset of points $Q_v$ associated with each $p_v$ are colored yellow and the special point $p^\star$ is colored blue. The edges denote the metric on these points. }
\label{fig:red}
\end{figure}

Let $n = |V|$.
We define the set of points $P$ in $\cI'$. For every edge $e \in E$, there is a point $p_e$ of capacity $0$. For every vertex $v \in V$, there is  a point $p_v$  with capacity $2d+1$. Further, for every point $p_v$, we associate a set $Q_v$ of $d$ new points, each with capacity 0.  Finally, there is  a special point $p^\star$ with capacity $\frac{n(d+1)}{3}+1$. 

In order to specify the metric on $P$, we first define an unweighted graph $H = (P, F)$, where $P$ is the vertex set and $F$ is the set of edges. For each $p_v$, we have an edge $(p_v, p^\star)$ We also have edges $(w,p_v), (w, p^\star)$ for each $w \in Q_v$). For each vertex $p_e$, where $e=(u,v)$ is an edge in $E$, we add edges $(p_e, p_u)$ and $(p_e, p_v)$ to the edge set $F$. 
 The distance between two nodes in $P$ is given by the shortest path distance in $H$. Finally, we set the parameter $k = \frac{2n}{3} + 1$. We now show that this reduction has the desired properties: 
\begin{lemma}
\label{lem:red1}
If the graph $G$ in the instance $\cI$ has a vertex cover of size at most $2n/3$, then there is a solution to $\cI'$ with objective function value at most $k$.
\end{lemma}
\begin{proof}
    Let $V'$ be a vertex cover of size exactly $2n/3$ in $G$ (we can add extra vertices to a vertex cover if its size is less than $2n/3$).   Consider a solution to $\cI'$ where we locate a center at each of the points $p_v, v \in V'$. Further, we locate a center at $p^\star$. Clearly, we have opened at most $k$ centers -- let $C$ denote the set of open centers. 

    We now specify the assignment of points in $P$ to the centers. First consider a point $p_e, e = (u,v)$. We know that at least one of $u$ and $v$ belong to $V'$, and hence, either $p_u$ or $p_v$ (or both) lie in $C$ -- assume wlog that $p_u \in C$. Then, we assign $p_e$ to $p_u$. For every $v \in V', $ we assign $p_v$ and all the points in $Q_v$ to $p_v$. All the remaining points are assigned to $p^\star$. 

    It is easy to check that each point in $P$ is assigned to a center which is joined to it by an edge in $H. $ Thus, the radius of each cluster is at most 1, and so the total cost of this solution is at most $k$. It remains to verify the capacity constraints. Every center $p_v \in C$ is assigned at 
    most $2d+1$ points because $|Q_v| = d,$ and the degree of $v$ is at most $d$. For the center $p^\star$, the set of points assigned to it (besides $p^\star$ itself) is given by $\bigcup_{v \in V \setminus V'} \left(\{v\} \cup Q_v \right).$ The size of this set is $(d+1)|V \setminus V'| = (d+1)n/3$. Thus, we satisfy the capacity constraint for $p^\star$ as well. This proves the desired result. 
\end{proof}

\begin{lemma}
\label{lem:red2}
If the graph $G$ in  the instance $\cI$ has a minimum vertex cover of size at least $(2n/3) \cdot (1 + \veps)$, then any solution to $\cI'$ has cost at least $k \cdot (1+\frac{\varepsilon}{2d})$.
\end{lemma}
\begin{proof}    
For sake of brevity, let $P_V$ denote the set of points $\{p_v : v \in V\}$ and $P_E$ denote $\{p_e: e \in E\}$. 
The only points with non-zero capacity are those in $P_V$  and the special node $p^\star$. Therefore, consider a solution obtained by choosing $p^\star$ and a subset $S$ of $2n/3$ points from $P_V$.  In the solution obtained above, let $S_0 \subseteq S$ by the set of points $p_v$ for which the corresponding cluster centered at $p_v$ has radius 0. Define $S_1, S_{\geq 2}$ similarly.  First observe that at least $2n\varepsilon/3$ points in $P_E$ are at a distance at least 2 from the set $S' := S \cup \{p^\star\}$. Indeed, there must be at least $2n\varepsilon/3$ edges which are not covered by $S$ (otherwise there is a vertex cover of size less than $2n(1+\varepsilon)/3$ in $G$). The points in $P_E$ corresponding to these edges must be assigned to a point in $S_{\geq 2}$ -- let $P_E'$ denote this susbet of $P_E$.

\begin{claim}
\label{cl:graph}
$|S_0| + \frac{2n \varepsilon}{3d} \leq |S_{\geq 2}|.$
\end{claim}
\begin{proof}

   We first observe  that 
   \begin{align}
    \label{eq:key}
    |S_0| \cdot d + \frac{n(d+1)}{3} + \frac{2n\veps}{3} \leq \frac{n(d+1)}{3} + |S_{\geq 2}| \cdot d
\end{align}

The l.h.s. of the above inequality counts the number of points which are covered by clusters of radius at least 2 or are assigned to $p^\star$. The term $|S_0|d$ counts the number of points in $Q_v$ for $v \in S_0$. The term $n(d+1)/3$ counts the points in $\{v\} \cup Q_v$ for all $v \in P_V \setminus S$. And the term $\frac{2n\veps}{3}$ accounts for the set of points in $P_E'$. 

Now consider the r.h.s. First it is not hard to see that for any point $v \in S_{\geq 2}$, the points in $Q_v$ should be assigned to $v$. This leaves a capacity of only $d$ for other points which can be assigned to such a point $v.$ Thus the r.h.s. above accounts for the total number of points (other than $Q_v, v \in S_{\geq 2}$) that can be assigned to $p^\star$ or $S_{\geq 2}$. The desired result now follows directly from~\eqref{eq:key}. 
\end{proof}

The cost of the solution $S \cup \{p^\star\}$ is at least $1 + |S_1| + 2 \cdot |S_{\geq 2}|$. Now, since $|S_0| + \frac{2n \varepsilon}{3d} \leq |S_{\geq 2}|$, we see that this cost is at least  
\[
 1 + |S_1| + |S_0| + \frac{2n\veps}{3d} + |S_{\geq 2}|
= 1 + \frac{2n}{3} + \frac{2n\veps}{3d} = k \cdot \left(1 + \frac{\veps}{2d}\right).
\]
\end{proof}
~\Cref{thm:eth} now follows from~\Cref{thm:eth1},~\Cref{lem:red1} and~\Cref{lem:red2}.